\documentclass[11pt,english,notitlepage]{article}

\usepackage[utf8]{inputenc}
\usepackage[T1]{fontenc}
\usepackage{lmodern}
\usepackage[DIV=11]{typearea} 

\usepackage{microtype}

\usepackage{amsmath,amsfonts,amsthm,amssymb}
\usepackage[dvipsnames]{xcolor}
\usepackage{enumitem}
\usepackage{caption}
\usepackage{placeins}
\usepackage{graphicx}
\usepackage{comment}
\usepackage{thm-restate}
\usepackage{fancybox}
\usepackage[ocgcolorlinks]{hyperref}

\DeclareUnicodeCharacter{221A}{$\sqrt{}$}
\DeclareUnicodeCharacter{2113}{$\ell$}

\usepackage{bm} 

\usepackage{nicefrac}

\usepackage[ruled,vlined, linesnumbered]{algorithm2e}

\usepackage{float}
\usepackage{cleveref}
\usepackage{url}

\usepackage[
	style=alphabetic,
    maxalphanames=4,
    minalphanames=4,
    maxcitenames=4,
    minnames=1,
    maxbibnames=20,
	backref=true,
	doi=true,
	url=true,
	backend=biber
]{biblatex}
\colorlet{DarkRed}{red!50!black}
\colorlet{DarkGreen}{green!50!black}
\colorlet{DarkBlue}{blue!50!black}

\hypersetup{
	linkcolor = DarkRed,
	citecolor = DarkGreen,
	urlcolor = DarkBlue,
	bookmarksnumbered = true,
	linktocpage = true
}

\newcommand{\polylog}{\text{ polylog }}

\newcommand{\roundConst}{2}
\newcommand{\wH}{\widehat{H}}

\newtheorem{theorem}{Theorem}[section]
\newtheorem{corollary}[theorem]{Corollary}
\newtheorem{lemma}[theorem]{Lemma}

\newtheorem{claim}[theorem]{Claim}
\newtheorem{fact}[theorem]{Fact}

\newtheorem{informal theorem}[theorem]{Informal Theorem}

\theoremstyle{definition}
\newtheorem{definition}[theorem]{Definition}

\newtheorem*{theorem*}{Theorem}
\newtheorem*{corollary*}{Corollary}
\newtheorem*{conjecture*}{Conjecture}
\newtheorem*{lemma*}{Lemma}
\newtheorem*{thm*}{Theorem}
\newtheorem*{prop*}{Proposition}
\newtheorem*{obs*}{Observation}
\newtheorem*{definition*}{Definition}

\newtheorem*{remark*}{Remark}
\newtheorem*{rec*}{Recommendation}

\usepackage{todonotes}

\newcommand{\ap}{{p} }
\newcommand{\hap}{\overline{p}}

\newcommand*\samethanks[1][\value{footnote}]{\footnotemark[#1]}

\allowdisplaybreaks[1]

\makeatletter
\g@addto@macro\bfseries{\boldmath}
\makeatother

\makeatletter
\g@addto@macro\mdseries{\unboldmath}
\g@addto@macro\normalfont{\unboldmath}
\g@addto@macro\rmfamily{\unboldmath}
\g@addto@macro\upshape{\unboldmath}
\makeatother

\DeclareCiteCommand{\citem}
    {}
    {\mkbibbrackets{\bibhyperref{\usebibmacro{postnote}}}}
    {\multicitedelim}
    {}

\renewcommand*{\multicitedelim}{\addcomma\space}

\AtEveryCitekey{\clearfield{note}}

\makeatletter
\AtBeginDocument{%
    \newlength{\temp@x}%
    \newlength{\temp@y}%
    \newlength{\temp@w}%
    \newlength{\temp@h}%
    \def\my@coords#1#2#3#4{%
      \setlength{\temp@x}{#1}%
      \setlength{\temp@y}{#2}%
      \setlength{\temp@w}{#3}%
      \setlength{\temp@h}{#4}%
      \adjustlengths{}%
      \my@pdfliteral{\strip@pt\temp@x\space\strip@pt\temp@y\space\strip@pt\temp@w\space\strip@pt\temp@h\space re}}%
    \ifpdf
      \typeout{In PDF mode}%
      \def\my@pdfliteral#1{\pdfliteral page{#1}}
      \def\adjustlengths{}%
    \fi
    \ifxetex
      \def\my@pdfliteral #1{}
      \def\adjustlengths{\setlength{\temp@h}{-\temp@h}\addtolength{\temp@y}{1in}\addtolength{\temp@x}{-1in}}%
    \fi%
    \def\Hy@colorlink#1{%
      \begingroup
        \ifHy@ocgcolorlinks
          \def\Hy@ocgcolor{#1}%
          \my@pdfliteral{q}%
          \my@pdfliteral{7 Tr}
        \else
          \HyColor@UseColor#1%
        \fi
    }%
    \def\Hy@endcolorlink{%
      \ifHy@ocgcolorlinks%
        \my@pdfliteral{/OC/OCPrint BDC}%
        \my@coords{0pt}{0pt}{\pdfpagewidth}{\pdfpageheight}%
        \my@pdfliteral{F}
        \my@pdfliteral{EMC/OC/OCView BDC}%
        \begingroup%
          \expandafter\HyColor@UseColor\Hy@ocgcolor%
          \my@coords{0pt}{0pt}{\pdfpagewidth}{\pdfpageheight}%
          \my@pdfliteral{F}
        \endgroup%
        \my@pdfliteral{EMC}%
        \my@pdfliteral{0 Tr}
        \my@pdfliteral{Q}%
      \fi
      \endgroup
    }%
}
\makeatother

\title{Deterministic Incremental APSP with Polylogarithmic Update Time and Stretch}
\date{}

\author{Sebastian Forster\thanks{Department of Computer Science, University of Salzburg. This work is supported by the Austrian Science Fund (FWF): P 32863-N. This project has received funding from the European Research Council (ERC) under the European Union's Horizon 2020 research and innovation programme (grant agreement No~947702).}\\
University of Salzburg
\and Yasamin Nazari \samethanks \\ University of Salzburg \and Maximilian Probst Gutenberg\thanks{The research leading to these results has received funding from the grant ``Algorithms and complexity for high-accuracy flows and convex optimization'' (no. 200021 204787) of the Swiss National Science Foundation.} \\ ETH Zurich}

\begin{document}
\maketitle
\begin{abstract}

    We provide the first \emph{deterministic} data structure that given a weighted undirected graph undergoing edge insertions, processes each update with \emph{polylogarithmic} amortized update time and answers queries for the distance between any pair of vertices in the current graph with a \emph{polylogarithmic} approximation in $O(\log \log n)$ time.
    
    Prior to this work, no data structure was known for partially dynamic graphs, i.e., graphs undergoing either edge insertions or deletions, with less than $n^{o(1)}$ update time except for dense graphs, even when allowing randomization against oblivious adversaries or considering only single-source distances.
\end{abstract}

\section{Introduction}




Partially dynamic algorithms for approximate shortest path problems have received considerable attention in recent years.
In the partially dynamic setting the input graph is undergoing either edge insertions (incremental setting) or edge deletions (decreme1.265ntal setting).
The focus on partially dynamic distance approximation algorithms, instead of fully dynamic ones allowing both types of updates, has three major reasons:
\begin{itemize}
    \item Fully dynamic maintenance of exact distances or small-stretch approximations is sometimes not possible with small update time under plausible hardness assumptions~\cite{Patrascu10,AbboudW14,HenzingerKNS15,BrandNS19}.
    \item Partially dynamic algorithms often serve as a ``stepping stone'' for fully dynamic algorithms~\cite{Thorup05,AbrahamCT14,ForsterGH21}.
    \item In several applications, partially dynamic algorithms that are deterministic or use randomization against an adaptive adversary\footnote{This means that the ``adversary'' creating the sequence of updates is \emph{adaptive} in the sense that it may react to the outputs of the algorithm. This type of adversary is called ``adaptive online adversary'' in the context of online algorithms~\cite{Ben-DavidBKTW94}. In contrast, an \emph{oblivious} adversary needs to choose its sequence of updates in advance, which guarantees probabilistic independence between the updates and the random choices made by the algorithm. Deterministic algorithms obviously work against an adaptive adversary.} can be used as subroutines for solving static problems~\cite{Madry10,ChuzhoyK19,Chuzhoy21,BernsteinGS21,chen2022maximum,chen2022simple}.
\end{itemize}

The research line of developing partially dynamic distance approximation algorithms against adaptive adversaries has been especially successful for \emph{undirected} graphs: in particular, deterministic incremental and decremental algorithms with almost optimal amortized update time of $ n^{o(1)} $ exist for the single-source shortest paths problem (SSSP) with stretch $ (1 + o(1)) $~\cite{BernsteinGS21} and for the all-pairs shortest paths problem (APSP) with stretch $ n^{o(1)} $~\cite{Chuzhoy21}.\footnote{As usual, $ n $ denotes the number of vertices and $ m $ denotes the maximum number of edges of the graph. In the introductory parts of this paper, we assume that edge weights are integer and polynomial in $ n $ when stating running time bounds.}
These efforts were leveraged for the following applications in static algorithms:
\begin{itemize}
\item $ (1 + o(1)) $-approximate minimum-cost flow in time $ m^{1+o(1)} $ \cite{BernsteinGS21} (using deterministic decremental $ (1 + o(1)) $-approximate SSSP)
\item $ n^{o(1)} $-approximate multicommodity flow in time $ m^{1+o(1)} $~\cite{Chuzhoy21} (using deterministic decremental $ n^{o(1)} $-approximate APSP)
\item Exact minimum cost flow in time $ m^{1+o(1)} $~\cite{chen2022maximum} (using randomized $ n^{o(1)} $-approximate APSP against adaptive adversaries on expander graphs)
\item Deterministic nearly-linear time constructions of light spanners (using deterministic incremental $ O(1) $ -approximate bounded distance APSP).
\end{itemize}

All known algorithms for partially dynamic approximate SSSP and APSP that are deterministic (or randomized against an adaptive adversary) suffer from an ``$ n^{o(1)} $-bottleneck'' in their update time with the only exception being the partially dynamic approximate SSSP algorithm of Bernstein and Chechik~\cite{BernsteinC16} that achieves polylogarithmic update time in very dense graphs with $ \tilde \Omega (n^2) $ edges.\footnote{In this paper, we use $ \tilde O (\cdot) $- and $ \tilde \Omega (\cdot) $-notation to suppress factors that are polylogarithmic in $ n $.} Even if we allowed randomization against an oblivious adversary, the $ n^{o(1)} $-bottleneck persists for sparse graphs~\cite{LackiN22}.

It is thus an intriguing and important open problem to design improved deterministic algorithms with polylogarithmic update time and polylogarithmic stretch in all density regimes that ideally work against an adaptive adversary to allow the use in static algorithms.
In particular, it has recently been shown that a deterministic decremental APSP algorithm with polylogarithmic stretch and polylogarithmic update time would imply a $ \tilde O (m) $-time algorithm for finding balanced sparse cuts~\cite{chen2022simple} and improve the state-of-the-art running time of the (exact) minimum cost flow problem~\cite{chen2022maximum}.

\paragraph{Our Result.}

In this paper, we present the first shortest path algorithm breaking the $ n^{o(1)} $ update-time barrier in \emph{all density regimes} for a partially dynamic distance problem with polylogarithmic stretch.
We give an incremental algorithm for maintaining a distance oracle with polylogarithmic stretch that has polylogarithmic amortized update time.

\begin{theorem}
There is a deterministic algorithm that, given an undirected graph with real edge weights in $ [ 1, W ] $ undergoing edge insertions, in total time $ O (m \log n \log \log n + n \log^6 (nW) \log \log n) $ over all updates maintains a distance oracle with polylogarithmic stretch and query time $ O (\log \log n) $, where $ n $ denotes the number of vertices $ m $ denotes the final number of edges of the graph.
\end{theorem}

Note that, while our stretch guarantee leaves some room for improvement in the exponent of the logarithm, there is evidence that a substantial improvement might not be possible without sacrificing the polylogarithmic update time:
Based on popular hardness assumptions concerning static 3SUM or static APSP, Abboud, Bringmann, Khoury, and Zamir \cite{AbboudBKZ22} recently showed that a constant-stretch distance oracle cannot be maintained with update time $ n^{o(1)} $.
Furthermore, a space bound (and thus also a total update time bound) of $ \Omega (n^{1+1/k}) $ for distance oracles with stretch $ 2k - 1 $ follows from Erdős's girth conjecture~\cite{Erdos63,ThorupZ05}.

It is worth noting that our techniques are very different from previous approaches that usually employ Even-Shiloach trees~\cite{EvenS81,King99} and constructions based on the Thorup-Zwick distance oracle~\cite{ThorupZ05}.
We further use no heavy algorithmic machinery and -- apart from a dynamic tree data structure -- our paper is self-contained.
We crucially use a hierarchical vertex sparsifier construction by Andoni, Stein, and Zhong~\cite{andoni2020parallel} that originally was developed in the context of distance approximation algorithms with polylogarithmic depth for the PRAM model.
The major technical challenge in employing this hierarchy in a dynamic setting is controlling the recourse -- i.e., the number of induced updates -- from the bottom to the top.
A standard ``layer-by-layer'' analysis would lead to an exponential blowup that would at best result in an $ n^{o(1)} $ overhead.
Instead, we perform several modifications to the vertex sparsifier hierarchy that allow a more controlled propagation in the algorithm that avoids such blowups. These modifications require a more entangled analysis over different levels.
Hence despite the relatively simple algorithm, our analysis is quite technical; see Section~\ref{sec:overview} for an overview of our technical ideas.

\paragraph{Prior Work.}

For the \emph{fully dynamic} all-pairs shortest paths problem, we can in principle distinguish the two regimes of update time $ \Omega (n) $ and update time $ o(n) $ (the ``sublinear'' regime). 
Most earlier works have focused on the first regime~\cite{King99,DemetrescuI06} and the state-of-the-art fully dynamic algorithms for APSP have an amortized update time of $ \tilde O (n^2) $ to maintain an exact solution~\cite{DemetrescuI04,Thorup04} or an amortized update time of $ m^{1 + o(1)} $ to maintain a $ (2 + o(1)) $-approximate solution in undirected graphs~\cite{Bernstein09,BernsteinGS21} (which can then be combined with a dynamic spanner algorithm~\cite{BaswanaKS12}).
Several approaches exist to obtain comparable worst-case update time~\cite{Thorup05,AbrahamCK17,BrandN19,GutenbergW20b,BrandFN22} and to obtain subquadratic (but still superlinear) update time at the cost of polynomial query time~\cite{Sankowski05,RodittyZ11,RodittyZ12,BrandNS19, bergamaschi2021new, karczmarz2022subquadratic}.

Work on the sublinear regime has been pioneered by Abraham, Chechik, and Talwar~\cite{AbrahamCT14} with a trade-off between stretch and amortized update time for unweighed, undirected graphs, allowing for example for a stretch of $ O (\log n) $ with amortized update time $ O (n^{1/2 + \delta}) $ (for any constant $ \delta $) in sparse graphs with $ O (n) $ edges.
A trade-off for weighted, undirected graphs allowing for even faster update time has been presented by Forster, Goranci, and Henzinger~\cite{ForsterGH21}, allowing for example for both subpolynomial stretch and subpolynomial amortized update time.
Both of these algorithms are randomized and correct against an oblivious adversary.

For the \emph{partially dynamic} all-pairs shortest paths problem, we can similarly distinguish between algorithms with total update time $ \tilde O (m n) $ (and above) and faster ``subcubic'' algorithms.
In the first regime, the state of the art is as follows:
deterministic exact all-pairs shortest paths can be maintained with total update time $ \tilde O (n^3) $ in unweighted, directed graphs~\cite{DemetrescuI06,BaswanaHS07,EvaldFGW21} and $ (1 + \epsilon)$-approximate all-pairs shortest paths can be maintained with total update time $ \tilde O (mn/\epsilon) $ in weighted, directed graphs \cite{RodittyZ12,Bernstein16,HenzingerKN16} against an oblivious adversary, and in total update time $\tilde O (mn^{4/3}/\epsilon^2)$ against an adaptive adversary \cite{KarczmarzL19, EvaldFGW21}.

Subcubic algorithms go beyond the ``$mn$'' barrier in undirected graphs either by increasing the multiplicative stretch or by allowing extra additive stretch.
In terms of purely multiplicative stretch, Chechik~\cite{Chechik18} presented an algorithm that for any integer $ k \geq 2 $ maintains a distance oracle of stretch $ (2 + o(1)) (k-1) $ with total update time $ m n^{1/k+o(1)} $, yielding in particular logarithmic stretch with total update time $ m^{1+o(1)} $.
This result was refined by Łącki and Nazari~\cite{LackiN22} to in particular improve the total update time to $ \tilde O ((m + n^{1+o(1)}) n^{1/k}) $.
Prior works were relevant only for dense graphs~\cite{BernsteinR11} or had ``exponentially growing'' stretch guarantees~\cite{HenzingerKN18,AbrahamCT14}.
Recently, a subcubic partially dynamic algorithm with stretch $ 2 + o(1) $ has been developed as well~\cite{DoryFNV23}.
All of these subcubic algorithms for multiplicative stretch are randomized and assume an oblivious adversary.
A deterministic incremental algorithm with several trade-offs between stretch and update and query time was developed by~\cite{ChenGHPS20}; in particular their algorithm can provide constant stretch and total update time $ m^{1+o(1)} $.
Deterministic partially dynamic algorithms allowing deletions have been developed by Chuzhoy and Saranurak~\cite{ChuzhoyS21} and by Chuzhoy~\cite{Chuzhoy21}, where the latter work provides polylogarithmic stretch in total update time $ O (m^{1+\delta}) $ for any constant $ \delta $ (see also~\cite{BernsteinGS21}).
The state-of-the-art algorithm with ``mixed'' stretch guarantee has a multiplicative stretch of $ (1+o(1)) $, an additive stretch of $ 2(k-1) $, and a total update time of $ (n^{2-1/k+o(1)} m^{1/k} $~\cite{DoryFNV23}.
Prior works considered only the case $ k = 2 $~\cite{HenzingerKN16,AbrahamC13} or had an ``exponentially growing'' stretch guarantee~\cite{HenzingerKN14}.
Again, all of these subcubic algorithms for ``mixed'' stretch are randomized and assume an oblivious adversary

\section{Overview}\label{sec:overview}
We start by reviewing the techniques from \cite{andoni2020parallel}. We then explain several challenges we face in the dynamic settings and the modifications we make to the construction 
to overcome these challenges.

\paragraph{Review of the static construction of \cite{andoni2020parallel}.} Our starting point is a distance oracle proposed by \cite{andoni2020parallel} that supports fast distance queries with polylogarithmic stretch. We can see this structure as a \textit{hierarchy of vertex sparsifiers}\footnote{In \cite{andoni2020parallel} what we call a vertex sparsifier is called a sub-emulator. They use a type of sub-emulator for building a low-hop emulator, i.e.~a graph that approximates the distances only using paths with $O(\log \log n)$-hops. We do not need a low-hop emulator and instead use subemulators/vertex sparsifiers for maintaining a distance oracle with small update and query time.}. For a given graph $G=(V,E)$, a vertex sparsifier is a graph $H$, where $V(H) \subseteq V(G)$ and each vertex $v \in V$ has a representative vertex $p(v)$ (called a pivot) such that the distance between vertices $u,v \in V$ can be approximated with the distance between $p(u)$ and $p(v)$ in $H$. A hierarchy of vertex sparsifiers allows us to compute approximate distances on subsequently smaller graphs in each level, while trading off computation time with the stretch.

Specifically, the algorithm of \cite{andoni2020parallel} creates graphs $H_1,\ldots, H_k$ in $k=O(\log \log n)$ levels as follows: at each level $i$ for a parameter $b_i$, they choose $V(H_{i+1}) \subseteq V(H_i)$ by subsampling a set of $\tilde{O}(|V(H_i)|/b_i)$ vertices. By choosing appropriate $b_i$ values increasing double exponentially in each level (in our case by setting $b_i=2^{(6/5)^i}$) we have that after $k=O(\log \log n)$ levels the number of remaining vertices is very small. Roughly speaking, each level of the hierarchy incurs an additional constant multiplicative factor in the stretch, which we denote by $\alpha$ and it then remains to observe that $\alpha^k = \polylog n$. 
We next describe the procedure by \cite{andoni2020parallel} to compute a vertex sparsifier~$H'$ for a graph~$H$ with some target parameter $b$ such that $|V(H')| \approx |V(H)| / b$. The hierarchy $H_1, \ldots, H_k$ is then obtained recursively by computing $H_{i+1}$ from $H_i$ by using the described procedure with target parameter $b_i$. 

The procedure initially samples each vertex from $V(H)$ with probability $1/b$ to form the vertex set of $H'$. They define the pivot $p(v)$ for any vertex $v \in V(H)$ to be the vertex in $V(H') \subseteq V(H)$ that is closest to $v$ in $H$ (we assume that distances are unique for simplicity in the overview). We define $pivotDist(v)$ to be the distance from $v$ to its pivot in $H$. Using standard hitting set arguments, one can then argue that there are at most $\tilde{O}(b)$ vertices inside the ball $B_H(v, pivotDist(v))$ centered at $v$ with radius $pivotDist(v)$. Having found the vertex set of $H'$, it remains to define the edge set. Two types of edges are added\footnote{Based on this definition, we may be introducing multi-edges to $H'$.} to $H'$:
\begin{itemize}
    \item Type 1 (Ball edges): For each $v \in V(H), u \in B_H(v, pivotDist(v))$, we have an edge $(p(u), p(v)) \in H'$ of weight $pivotDist(v)+ dist_{H}(u,v)+pivotDist(u)$, and
    \item Type 2 (Projected edges): For each $e = (x,y) \in E(H)$, an edge $(p(x), p(y))$ of weight $pivotDist(x) + w_{H}(x,y)+ pivotDist(y)$ is added to $H'$.
\end{itemize}
Intuitively, the first type of edges connect vertices in $H$ if one vertex appears in the ball of the other and the second type of edges connect the boundaries of the balls.

\begin{figure}
    \centering
    \includegraphics[scale=0.6]{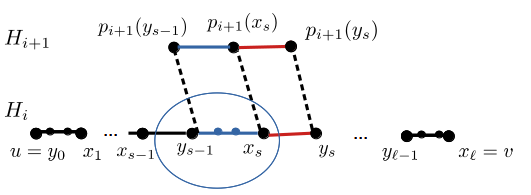}
    \caption{\small In the construction of \cite{andoni2020parallel} paths in $H_i$ are approximated by paths in $H_{i+1}$ going through level $i$ pivots (denoted by the function $p_{i+1}$) of certain vertices on the path. Here $x_s$ is in the ball of $y_{s-1}$ (represented by the circle). Dashed lines represent projections to the next level.}
    \label{fig:ASZpaths}
\end{figure}

We briefly sketch the stretch analysis: Consider vertices $u,v \in V(H)$, and let $\pi$ be the shortest path between $u$ and $v$ in $H$. We divide this path into segments defined by a sequence of vertices $u=y_0,...,y_{\ell-1}$ and $x_1,...,x_{\ell}=v$ defined as follows: Starting from $y_0 := u$, for each $s >0$ let $x_{s+1}$ be the last vertex on $\pi$ such that $x_{s+1} \in B_H(y_s, pivotDist(y_s))$. Then $y_{s+1}$ is set to the next vertex on $\pi$ right after $x_{s+1}$. We stop when $x_{s+1} = v$. Essentially, we segment the path $\pi$ to alternately take maximal segments contained in balls to pivot distances, and using the original edge (see also \Cref{fig:ASZpaths}). \cite{andoni2020parallel} then suggest that this path can simply be projected to $H'$ by taking the path $\langle p(u = y_0), p(x_{1}), p(y_1), p(x_2), \ldots, p(x_{\ell} = v) \rangle$ in $H'$ where every two vertices are connected by an edge in $H'$ as can be verified from the procedure above. 

One can then show that $dist_{H'}(p(y_s), p(y_{s+1})) \approx dist_{H}(y_s, y_{s+1}) + pivotDist(y_s) + pivotDist(y_{s+1})$ (here $\approx$ hides a constant factor). But it is not hard to see that for any $s$, we have $dist_H(y_s, y_{s+1}) > pivotDist(y_s)$. Summing over path segments, we thus get that $dist_{H'}(p(u), p(v)) \approx dist_{H}(u, v) +  pivotDist(v)$. 

Unfortunately, one cannot hope to get rid of an additive term scaling linearly in $pivotDist(v)$ in the approximation as can be seen from straight-forward worst-case examples. However, one can use classic distance oracle query techniques: For query pair $u,v$, we either have that $u \in B_H(v, pivotDist(v))$ and keeping these balls and the respective distances explicitly in a dictionary, one can then return the exact distance. Otherwise, we have $pivotDist(v) < dist_H(u,v)$ and therefore $dist_{H'}(p(u), p(v)) \approx dist_{H}(u, v)$.

Finally, it is easy to see that the sampling ensures $|E(H')| \leq \tilde{O}(|V(H)| b) + |E(H)|$. 

Putting this result back to our hierarchy $H_1, H_2, \ldots, H_k$, we have that one can straight-forwardly query the distances between any two vertices in the original graph $G$ in time $O(k)$ by applying the discussed query procedure iteratively. Further, for each $i$, we have $|E(H_i)| \leq |E(G)| + \sum_i O(|V(H_i)|/b_i)$ where $b_i$'s are chosen carefully to ensure $|E(H_i)| \leq m + \tilde O(n)$.



\begin{algorithm}[h]
\DontPrintSemicolon


    \While{$\exists v \in V(H)$ such that $|B_{H}(v, \frac{1}{4}\widetilde{pivotDist}(v))| \geq {b}$}{
        
        Let $B_v$ be a set of size ${b}$ such that $B_v \subseteq B_{H}(v, \frac{1}{4}\widetilde{pivotDist}(v))$.\\
     
        \If(\label{lne:ifCaseMainUpdate}){$\exists u \in B_v$ with $\widetilde{pivotDist}(u) < \frac{1}{2} \widetilde{pivotDist}(v)$}{
            $p(v) \gets p(u)$; $\widetilde{pivotDist}(v) \gets dist_{H}(v,u) + \widetilde{pivotDist}(u)$.
        }\Else{
            \lForEach{$u \in B_{v}$}{
                $p(u) \gets v$; $\widetilde{pivotDist}(u) \gets dist_{H}(u,v)$.
            }
        }
    }
\caption{$\textsc{UpdateApproxPivots}(H, b)$}
\label{alg:overview}
\end{algorithm}

\paragraph{Incremental algorithm for one level.} 
We give an overview of our algorithm for maintaining these two types of edges for a single level before describing the modifications needed for making the algorithm efficient over all levels.

The first obstruction to maintaining the vertex sparsifier of the last section in incremental settings is that for each vertex $v \in V(H)$, even if the pivot $p(v)$ does not change, the pivot distance $pivotDist(v)$ might change after almost each of the $ m $ insertions. While on average most vertices might only undergo few pivot distance changes, we still might have some node $ v $ of large degree and would have to adjust the weight of projected edges (type 2) incident on $ v $ with every change in $ v $'s pivot distance. To avoid such a running time overhead for simply maintaining the weights of projected edges, we maintain an approximation $\widetilde{pivotDist}(v)$ of $pivotDist(v)$ such that whenever $\widetilde{pivotDist}(v)$ changes, it decreases by a constant factor (thus the total number of changes is $O(\log(nW))$. We maintain the pivot $p(v)$ to be some vertex in $V(H')$ that is roughly at distance  $\widetilde{pivotDist}(v)$ (in our case, all approximations are within a factor 4 of each other). In the following, we merely describe an algorithm to maintain the approximate pivots and the corresponding distance estimates as it is straight-forward to maintain the edge set of $H'$ from this information.

We give our procedure in \Cref{alg:overview}. We assume here that all edge weights in $H$ are powers of $2$ which is w.l.o.g. since we only want to obtain constant stretch. Here, we skip the initialization procedure for brevity. The algorithm is then invoked after every edge insertion to $H$. The algorithm works as follows: whenever it detects that the ball $B_{H}(v, \frac{1}{4}\widetilde{pivotDist}(v))$ contains more than $b$ vertices, it checks for a closer pivot for $v$ which then decreases $\widetilde{pivotDist}(v)$ significantly. Therefore, the algorithm searches over $b$ vertices in $B_{H}(v, \frac{1}{4}\widetilde{pivotDist}(v))$ and asks them whether their pivot would make a good candidate. If such a candidate is found, it becomes the new pivot of $v$. Otherwise, we make $v$ a pivot itself and assign it to the set of vertices scanned as a pivot (the vertices in $B_v$). In this latter case, each vertex in $B_v$ has its (approximate) pivot distance decreased significantly. We maintain the vertex set $V(H')$ to be the image of all pivot functions $p$ from the current and previous stages. 

To implement the while-loop in \Cref{alg:overview}, we use a truncated Dijkstra's algorithm from each vertex $v$ to explore $B_{H}(v, \frac{1}{4}\widetilde{pivotDist}(v))$, however, we abort the procedure after seeing $b$ vertices. Using adjacency lists sorted by weight, we can implement this procedure in time $\tilde{O}(b^2)$. Note that in between any two stages, for a vertex $v$, if no edge/a multi-edge of equal or higher weight is inserted into $B_{H}(v, \frac{1}{4}\widetilde{pivotDist}(v))$, we can simply ignore the update and do not need to recompute. But in the other case, for fixed $\widetilde{pivotDist}(v)$, we have that $B_{H}(v, \frac{1}{4}\widetilde{pivotDist}(v))$ is increasing over time, and the number of edges with different weights and endpoints in the final ball (before exceeding $b$ vertices) is $b^2 \log(nW)$ edges. Hence we have that there are at most $O(b^2 \log nW)$ recomputations before $\widetilde{pivotDist}(v)$ changes. Overall this incurs total time  $\tilde{O}(b^4 \log^2 nW)$ per vertex $v \in V(H)$, and thus $\tilde{O}(m + |V(H)| b^4 \polylog nW)$ overall.\footnote{Here, we were slightly imprecise but in the exact analysis one has a higher power for the $\log(nW)$ factor.}

Finally, observe that $|V(H')|$ is only increased if we enter the else-case. But in this case, $b$ vertices have their pivot distance significantly decreased. We can therefore upper bound the number of vertices in $H'$ by $O(|V(H)|\log(nW)/b )$. However as we see next, bounding the recourse on the number of edges will be problematic over all the $k$ levels.

\paragraph{Challenges in maintaining Projected Edges.}
The naive approach for maintaining the vertex-sparsifier hierarchy would be to run the aforementioned algorithm for each $ 1 \leq i < k $ in a black-box manner to maintain $ H_{i+1} $ as a vertex sparsifier of $ H_i $. 
In particular, the edges added to $ H_{i+1} $ over the course of the algorithm appear as insertions to the algorithm maintaining the vertex sparsifier $ H_{i+2} $ of the next level.
Thus $O(|E(H_i)| \log nW + |V(H_i)| b_i^2 \log(n W) )$ edges are inserted to $H_{i+1}$ in total: $O(|V(H_i)|b^2_i  \log(n W))$ type 1 (ball) edges, and $O(|E(H_i)| \log nW)$ type~2 (projected) edges. 

Starting from $G=H_1$ with $m$ edges, this naive approach with $k=O(\log \log n)$ levels leads to a bound of at least $ m \cdot O(\log nW)^{\log \log n}$ type 2 edges inserted to the top level, each of which needs \emph{at least} constant time to be processed.
Therefore, the black-box approach will not give us the desired $\tilde{O}(m)$ total update time. Instead, we propose a more careful approach for avoiding the exponential blow up in the number of inserted edges within the hierarchy that we explain next.

\paragraph{Maintaining the hierarchy via multi-level projections.}

The challenge discussed means that we cannot afford to have a chain of projections from lower levels to higher levels.
In the following, $ p_{i+1} (v) $ denotes the pivot of some node $ v \in V (H_i) $ maintained by the incremental algorithm at level $ i $ and $ \widetilde{pivotDist}_{i+1} (v) $ denotes its corresponding approximate pivot distance.
For the sake of concreteness, consider some edge $ (u, v) $ in $ H_1 = G $ and the pivots $ p_2 (u) \in V (H_2) $ and $ p_2 (v) \in V (H_2) $ of its endpoints.
Following our previous definition of $ H_2 $ and our process for maintaining approximate pivot distances, there would be a projected edge $ (p_2(u), p_2(v)) $ in $ H_2 $ of weight: 
\begin{equation*}
w_{H_2} (p_2(u), p_2(v)) = \widetilde{pivotDist}_2 (u) + w_{H_1} (u,v) + \widetilde{pivotDist}_2 (v) \, .
\end{equation*}
This edge would in turn be projected to $ H_3 $ by an edge $ (p_3(p_2(u)), p_3(p_2(v))) $ of weight
\begin{align*}
w_{H_3} (p_3(p_2(u)), p_3(p_2(v))) &= \widetilde{pivotDist}_3 (p_2(u)) + w_{H_2} (p_2(u), p_2(v)) + \widetilde{pivotDist}_3 (p_2(v)) \\
&= \widetilde{pivotDist}_3 (p_2(u)) + \widetilde{pivotDist}_2 (u) + w_{H_1} (u,v) \\ &\hspace{1em} + \widetilde{pivotDist}_2 (v) + \widetilde{pivotDist}_3 (p_2(v)) \, .
\end{align*}
As explained above, the ``black box'' approach would mean to insert projections of $ (u, v) $ to $ H_3 $ whenever $ p_3(p_2(u)) $ changes, which happens $ O (\log (nW)) $ times for each of the $ O (\log (nW)) $ choices of $ p_2(u) $.
This bounds the number of insertions of projections of $ (u, v) $ to $ H_3 $ by $ O (\log (nW)^2) $ (and in general the number of insertions to $ H_i $ by $ O (\log (nW))^{i-1} $).

Our main idea for obtaining a better bound is to employ another lazy updating scheme: we insert a projection of $ (u, v) $ to $ H_3 $ only when the sum above determining the edge weight changes significantly, in particular whenever the ``left part'' $ \widetilde{pivotDist}_3 (p_2(u)) + \widetilde{pivotDist}_2 (u) $ or the ``right part'' $ \widetilde{pivotDist}_2 (v) + \widetilde{pivotDist}_3 (p_2(v)) $ decreases by a constant factor.
In this way, we ``reproject'' $ (u, v) $ to $ H_3 $ only $ O (\log (nW)) $ times, a bound that is independent on the level at which the projection happens, which gives us the desired control in the number of insertions at each level.
Note that such projections to higher levels are not only carried out for the edges of $ G $, but also for the type~1 (ball) edges introduced at each level of the hierarchy, which can be done analogously.

More precisely, we define a set of \textit{base edges}, which are intuitively the level $i$ edges that were not previously projected from a lower level. To this end, it is convenient to define $p_i(u) = p_i( \ldots p_2(p_1(u)) \ldots)$ for any vertex $u \in V$. Then at level $i+1$, we add, in the lazy fashion explained above, a \textit{projected edge} $(p_{i+1}(u), p_{i+1}(v)) \in H_{i+1}$ corresponding to each base edge $(u,v) \in E(H_j)$ from level $j \leq i$ and setting the weight (at time of projection) to be $\sum_{j\leq i}  \widetilde{pivotDist}_j(u) + w_{H_j}(u,v)+ \sum_{j\leq i} \widetilde{pivotDist}_j (v)$.

We show that we can carry out our idea efficiently by utilizing a dynamic tree data structure on the forest induced by connecting each vertex of the hierarchy to its pivot at the next level (weighted by approximate pivot distance).
Whenever for some vertex $ v $ in some $ H_j $ the sum of the approximate pivot distances along the tree path to its ancestor pivot $ v' $ at some level $ i > j $ decreases by a constant factor, we insert to $ H_i $ the projections of all (non-projected) edges incident on $ v $ in $ H_j $.
We call the corresponding pivot $ v' $ a significantly improving pivot of $ v $ at level $ i $.
These significantly improving pivots will play a major role in our algorithm, as we explain next.


\paragraph{Challenges introduced by considering significantly improving pivots.}

While we have abandoned the ``black-box'' level-by-level approach for efficiency reasons, we still want to, in spirit, follow the proof strategy of~\cite{andoni2020parallel}, which is an inductive level-by-level stretch analysis.
This ``mismatch'' causes certain issues.
Consider again the argument of~\cite{andoni2020parallel} to show that any shortest path~$ \pi $ in $ H_i $ has a suitable approximation in $ H_{i+1} $ (see \Cref{fig:ASZpaths}).
The path~$ \pi $ is divided into segments and each segment is represented by an edge in $ H_{i+1} $.
In particular, some of these segments consist of single edges $ (x_s, y_s) $, which in particular are type~2 edges in $ H_ i $.
In the original proof, $ H_{i+1} $ contains the projection $ (p_{i+1} (x_s), p_{i+1} (y_s)) $  of $ (x_s, y_s) $ (where $ p_{i+1} (x_s) $ and $ p_{i+1} (y_s) $ are the current pivots of $ x_s $ and $ y_s $, respectively).

However, after our modifications for lazy updating we only have the weaker guarantee that $ (x_s, y_s) $ was inserted previously as the projection of some (non-projected) edge $ (\bar{x}_s, \bar{y}_s) $ from some lower level $ j < i $.
Additionally we know that $ H_{i+1} $ contains the projection $ e $ of the edge $ (\bar{x}_s, \bar{y}_s) $ from $ H_j $ and the endpoints of $ e $ are the the last significantly improving pivots at level $ i+1 $ of $ \bar{x}_s $ and $ \bar{y}_s $, respectively (see \Cref{fig:secondEdgeLifts} in \Cref{sec:main}).
The major challenge now is to still find a suitable path from $ p_{i+1} (x_s) $ to  $ p_{i+1} (y_s) $ in $ H_{i+1} $, which should include $ e $ to somehow relate the length of this path to the weight of $ (x_s, y_s) $.

\paragraph{New edges for significantly improving pivots.}

We address this challenge by introducing two new types of edges (with appropriately chosen weights) into our vertex sparsifiers:
The first new type gives us an edge from the current pivot of $ x_s $ (i.e., $ p_{i+1} (x_s) $) to the last significantly improving pivot of $ x_s $.
The second new type gives us an edge from the last significantly improving pivot of $ x_s $ to the last significantly improving pivot of $ \bar{x}_s $, i.e., the first endpoint of $ e $.
Similarly, we can use the new types of edges to find a path from the second endpoint of $ e $ to the current pivot of $ y_s $ (i.e., $ p_{i+1} (y_s) $), and thus find the desired path from $ p_{i+1} (x_s) $ to $ p_{i+1} (y_s) $.
Since the new types of edges are used in a somewhat special configuration, we can argue that they can be included in the hierarchy with only polylogarithmic overheads.
Setting the edge weights appropriately to obtain a stretch bound for this path in $ H_{i+1} $ requires some intricate estimates.
The exact definition of these edges and the full analysis can be found in~\Cref{sec:main}.

\section{Preliminaries}\label{sec:preliminaries}

\paragraph{Basic Notation.} For a general (multi-)graph $H$, we denote the edge set of the graph by $E(H)$, its vertex set by $V(H)$ and its weight function by $w_H$ where $w_H$ maps each edge in $E(H)$ to a positive number. We denote the distance between any two vertices $u,v \in V(H)$ in the graph $H$ by $dist_H(u,v)$. We denote by $B_H(u, r) = \{ dist_H(u,v) \leq r\}$ the ball at $u$ in $H$ of radius $r$. We say that $H$ is incremental if it is undergoing edge insertions. 

In this article, we denote by $G=(V,E,w)$ the input graph and define $n := |V|$, $m = |E|$ and let $w$ be the weight function with image in $[1, W]$.

\paragraph{Encoding of the Adjacency List.} We assume that additional to the usual encoding, we have for each (multi-)graph $H$ an adjacency list for each vertex $v$ denoted by $\textsc{Adj}_{H,v}$ stored as a doubly-linked list where the edges incident to $v$ appear sorted lexicographically first by weights and then by time of arrival. Here we define time of arrival for an edge to be equal to the number of edges that were in the graph before the edge was added where we assume without loss of generality that edges are added one after another and the initial graph $H$ is empty. We often index the adjacency list like an array and use $\textsc{Adj}_{H,v}[1, b]$ to refer to the set of the first $b$ edges in the adjacency list of $v$ (i.e. the $b$ edges of smallest weight). 

\paragraph{Update time.} The \emph{total update time} of an incremental algorithm is (a bound on) the sum of the running times spent by the algorithm for processing all of the $ m $ insertions and its \emph{amortized update time} is its total update time divided by $ m $.\footnote{Similarly, the total update time of a decremental algorithm is usually the sum of the running times spent by the algorithm for processing up to $ m $ deletions in a graph with initially $ m $ edges.}

\paragraph{Miscellaneous.} We define $\lceil x \rceil_{\roundConst} = \lceil x/2 \rceil \cdot 2$, where we round up $ x $ to the next multiple of~$ 2 $.

We refer to the $t$-th \textit{stage} of a dynamic algorithm as the instructions it performs after the $t$-th update.
We refer to the value of a variable or function at stage $ t $ as the value directly after the $t$-th stage and write it with the superscript ``$(t)$'' ; $ p^{(t)} (v) $ for example denotes the pivot of $ v $ at stage $ t $. We omit the superscript when it is clear from the context, for instance when we talk about the current stage.

\section{Full Algorithm and Analysis}
We start by giving the hierarchy that we maintain. We then give an algorithm to maintain the hierarchy efficiently that allows for additional query access. Finally, we give the query algorithm.

\subsection{A Distance-Preserving Vertex Sparsifier Hierarchy}\label{sec:main}

\begin{definition}[Distance-Preserving Vertex Sparsifier Hierarchy]\label{def:sparsifier hierarchy}
Given an incremental, undirected, weighted graph $G = (V,E,w)$, a $k$-level hierarchy maintaining algorithm is an algorithm that maintains vertex sparsifiers $H_1, H_2, \ldots, H_k$ for some positive integer $k$, with $V(H_1) \supseteq V(H_2) \supseteq \ldots \supseteq V(H_k) \neq \emptyset$ where $H_1 = G$ and for every $1 \leq i \leq k$, $H_i$ is an \emph{incremental} graph (with vertex insertions). We have a pivot function set to $p_1(v)=v$ for the initial level. The algorithm maintains for every $1 \leq i < k$:
\begin{enumerate}
    \item \label{prop:apprxDist}\label{prop:apAreCloseToPivots} an approximate pivot function $p_{i+1} : V \mapsto V(H_{i+1})$ that acts as the identity on $V(H_{i+1})$, and an estimator of the distance from each $v \in V(H_i)$ to its approximate pivot $\widetilde{pivotDist}_{i+1}(v)$. We enforce that
    \[
        dist_{H_i}(v, V(H_{i+1})) \leq dist_{H_i}(v, \ap_{i+1}(v)) \leq \widetilde{pivotDist}_{i+1}(v) \leq 4 \cdot dist_{H_i}(v, V(H_{i+1})).
    \]
    It also maintains for each $v \in V$, the quantity $\widetilde{pivotDist}_{i+1}(v) = \sum_{j \leq i} \widetilde{pivotDist}_{j+1}(p_j(v))$ and the value $\widetilde{minPivotDist}_{i+1}(v) = \min_{t' \leq t}\; \widetilde{pivotDist}^{(t')}_{i+1}(v)$ where $t$ is the current stage of the graph.
    For each $v \in V \setminus V(H_{i+1})$, we have $p_{i+1}(v) = p_{i+1}(p_i(v))$. 
    \item It further maintains for each $v \in V$, the \emph{last (significantly) improving pivot} $\overline{p}_{i+1}(v)$ that we define to be the approximate pivot $\ap^{(t')}_{i+1}(v)$ for $t' =  \min \{ t'' \;|\; \widetilde{pivotDist}^{(t'')}_{i+1}(v) \leq  \lceil \widetilde{minPivotDist}_{i+1}(v) \rceil_{\roundConst} \}$.\label{lastimprPivot}
\end{enumerate}

Given these values, our algorithm maintains each $H_{i+1}$ as an incremental graph consisting of two types of edges in $E(V_{i+1})$: \textit{base edges} $E^{base}_{i+1}$ which are the edges first introduced in level $i+1$, and \textit{projected edges} $E^{proj}_{i+1}$ which are projected to level $i+1$ from lower level graphs. For convenience, we define the sets $E^{base}_1 = E$ and $E^{proj}_1 = \emptyset$.\\
The algorithm is required to maintain a set of base edges $E^{base}_{i+1}$ which contains
\begin{enumerate}
   \setcounter{enumi}{2}
     \item for each $u \in V(H_i)$, and $v \in B_{H_i}(u, \frac{1}{4} \cdot \widetilde{pivotDist}_{i+1}(u))$,  an edge $(p_{i+1}(u), p_{i+1}(v))$ in $E^{base}_{i+1}$ with weight $8 \cdot \lceil \widetilde{pivotDist}_{i+1}(u) \rceil_{\roundConst}$. 
     \label{iball_edges} 
     
     \item For any vertex $v \in V(H_i)$, let $0 = t_1 < t_2 < \ldots <t_h \leq t$ be such that for $j \geq 1$, we have $t_{j+1}$ to be the first stage after stage $t_j$ such that $\ap^{(t_{j+1})}_{i+1}(v) \neq \ap^{(t_{j+1} - 1)}_{i+1}(v)$. Then, we have for any $1 \leq j < \ell \leq h$, a base edge $(\ap ^{(t_j)}_{i+1}(v), \ap^{(t_{\ell})}_{i+1}(v)) \in E^{base}_{i+1}$ with weight $8 \cdot \lceil \widetilde{minPivotDist}^{(t_j)}_{i+1}(v) \rceil_{\roundConst}$. \label{item:glueEdgesBetweenAPs} 
    
    \item For any vertex $v \in V$, times $t' \leq t$, we have at stage $t$, for $x = \hap^{(t')}_{i}(v)$, an edge $(\hap^{(t)}_{i+1}(x), \hap_{i+1}^{(t)}(v))$ in $E^{base}_{i+1}$ of weight $\lceil \widetilde{minPivotDist}^{(t)}_{i+1}(x) \rceil_2 + \lceil \widetilde{minPivotDist}^{(t')}_i(v) \rceil_2 + \lceil \widetilde{minPivotDist}^{(t)}_{i+1}(v) \rceil_2$.
      \label{item:glueEdgesBetweenBaseConnectors}
\end{enumerate}
Additionally, the algorithm maintains a set of projected edges $E^{proj}_{i+1}$ which contains 
\begin{enumerate}
    \setcounter{enumi}{5}
    \item for $j \leq i$ and $e = (x,y) \in E^{base}_j$, the edge $(\hap_{i+1}(x), \hap_{i+1}(y))$ of weight $ \lceil \widetilde{minPivotDist}_{i+1}(x)\rceil_{\roundConst} + \lceil w_{H_j}(e) \rceil_{\roundConst} +  \lceil \widetilde{minPivotDist}_{i+1}(y)\rceil_{\roundConst}$ in $E^{proj}_{i+1}$. \label{projected_edges}
\end{enumerate}
We also set for all $v \in V$ $\bar{p}_1(v)=p_1(v)=v$.
\end{definition}

We first establish the following simple facts that prove useful in the next proof of the main theorem of this section.

\begin{fact}\label{fact:pivotIsRoughlyMinPivot}
For any $v \in V(H_i)$, we have $dist_{H_i}(v, \ap_{i+1}(v)) \leq 4 \cdot \widetilde{minPivotDist}_{i+1}(v)$.
\end{fact}
\begin{proof}
We have at any stage $t$, $\widetilde{minPivotDist}_{i+1}(v) = \min_{t'' \leq t} \widetilde{pivotDist}_{i+1}^{(t'')}(v)$ and let $t'= \text{argmin}_{t''\leq t} \widetilde{pivotDist}_{i+1}^{(t'')}(v)$. Using that $v \in V(H_i)$, \Cref{prop:apprxDist} and that $H_i$ is incremental, we have that 
\begin{align*}
dist_{H_i}(v, \ap_{i+1}(v)) &\leq 4 \cdot dist_{H_i}(v, V(H_{i+1})) \leq 4 \cdot dist_{H^{(t')}_i}(v, V(H^{(t')}_{i+1})) \\
&\leq 4 \cdot \widetilde{pivotDist}_{i+1}^{(t')}(v) \leq 4 \cdot \widetilde{minPivotDist}_{i+1}(v).
\end{align*}
\end{proof}

We can now prove the main result of this section: we show that any algorithm that maintains a hierarchy $H_1, H_2, \ldots, H_k$ as described in \Cref{def:sparsifier hierarchy} has distances in $H_{i+1}$ being constant-factor approximations of distances in $H_i$. 

\begin{theorem}\label{thm:mainApproxTheorem}
Given a $k$-level hierarchy maintaining algorithm as described in \Cref{def:sparsifier hierarchy}, for any $1 \leq i \leq k -1$, we have for any $u,v \in V(H_{i+1})$, that $dist_{G}(u,v) \leq dist_{H_{i+1}}(u,v) \leq 3629 \cdot dist_{H_i}(u,v)$. 
\end{theorem}

We defer the proof of the lower bound given in the theorem  to \Cref{sec:LowerBoundingVertexSparsifier} as the proof is rather tedious, and focus for the rest of the section on achieving the upper bound.

\paragraph{Creating Path Segments.} Let $\pi$ be the shortest path between $u$ and $v$ in $H_i$. We show that there is a path with the desired stretch in $H_{i+1}$ by dividing this path into segments defined by a sequence of vertices $u=y_0,...,y_{\ell-1}$ and $x_1,...,x_{\ell}=v$ on $\pi$ found by the following procedure:
\begin{itemize}
    \item $y_0 \gets u$, $s \gets 0$, repeat the following two steps:
    \item Let $x_{s+1}$ be the last vertex on $\pi$ such that $dist_{H_i}(y_{s}, x_{s+1}) \leq \frac{1}{4}\widetilde{pivotDist}_{i+1}(y_{s})$. If $x_{s+1} = v$, the procedure terminates.
    \item Otherwise, find $y_{s+1}$ to be the vertex that appears next on $\pi$ after $x_{s+1}$. $s \gets s+1$.
\end{itemize}
    \paragraph{Mapping Path Segments into $H_{i+1}$.} For $s = 0, \ldots, \ell - 1$, we have by \Cref{iball_edges} that edge $(\ap_{i+1}(y_s), \ap_{i+1}(x_{s+1}))$ exists in $H_{i+1}$ of weight $8 \cdot \lceil \widetilde{pivotDist}_{i+1}(y_s) \rceil_{\roundConst}$. For $s = 1, \ldots, \ell-1$, we need to find paths from  $\ap_{i+1}(x_s)$ to $\ap_{i+1}(y_s)$. However, this turns out to be a considerably more laborious task. The following lemma summarizes the result.

\begin{figure}
    \centering
    \includegraphics[scale=0.4]{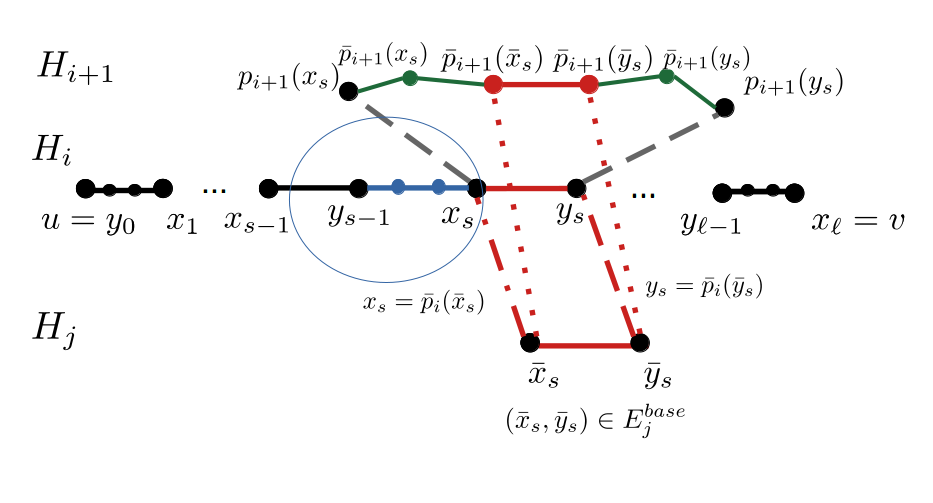}
    \caption{A sketch of the paths discussed in \Cref{lma:secondtypeEdgesLift} and the projections from level $i$ to $i+1$ and also from a level $j <i$ to $i+1$. Note that here in applying \Cref{clm:simpleHelperStretch} we set $z= \{x_s, y_s\}$. Dashed lines represent projections from lower levels to higher levels. Here $x_s$ is in the ball of $y_{s-1}$ and we have $(x_s,y_s) \in H_i$.}
    \label{fig:secondEdgeLifts}
\end{figure}

\begin{lemma}\label{lma:secondtypeEdgesLift}
For $s = 1, \ldots, \ell-1$, we have $dist_{H_{i+1}}(\ap_{i+1}(x_s), \ap_{i+1}(y_s)) \leq  69 \cdot w_{H_i}(x_s, y_s) + \sum_{z \in \{x_s, y_s\}} 85 \cdot \lceil \widetilde{minPivotDist}_{i+1}(z)\rceil_{\roundConst}.$
\end{lemma}
\begin{proof}
Since $x_s$ and $y_s$ are neighbors in $H_i$
we either have $(x_s, y_s) \in E^{base}_i$ or $(x_s, y_s) \in E^{proj}_i$. In the following argument both we handle both of these cases simultaneously. 
By construction of the path segments, and by \Cref{projected_edges}, that there is an edge $(\overline{x_s}, \overline{y_s}) \in E^{base}_j$ for some $j \leq i$, such that $(\hap^{(t')}_{i}(\overline{x_s}), \hap^{(t')}_i(\overline{y_s})) = (x_s, y_s)$ at some stage $t' \leq t$. Note that here if $(x_s,y_s) \in E^{base}_i$ then we are using the fact that $\bar{x}_s=x_s$ and $\bar{y}_s=y_s$.

Again by \Cref{projected_edges}, we have $(\hap_{i+1}(\overline{x_s}), \hap_{i+1}(\overline{y_s})) \in E^{proj}_{i+1}$ of weight $\lceil \widetilde{minPivotDist}_{i+1}(\overline{x_s})\rceil_{\roundConst} + \lceil w_{H_j}(\overline{x_s}, \overline{y_s}) \rceil_{\roundConst} +  \lceil \widetilde{minPivotDist}_{i+1}(\overline{y_s})\rceil_{\roundConst}$. 
    
It remains to find paths from $\ap_{i+1}(x_s)$ to $\hap_{i+1}(\overline{x_s})$ and from $\ap_{i+1}(\overline{y_s})$ to $\ap_{i+1}(y_s)$. To this end, we employ the simple claim below. This establishes the existence of a path $\ap_{i+1}(x_s) \leadsto \hap_{i+1}(\overline{x_s}) \leadsto \hap_{i+1}(\overline{y_s}) \leadsto \ap_{i+1}(y_s)$. 

\begin{claim}\label{clm:simpleHelperStretch}
For any $v \in V$ and $t' \leq t$, where we define $z = \hap_i^{(t')}(v)$, $dist_{H_{i+1}}(\ap^{(t)}_{i+1}(z), \hap_{i+1}^{(t)}(v)) \leq \lceil \widetilde{minPivotDist}^{(t')}_i(v) \rceil_2 + 17 \lceil \widetilde{minPivotDist}^{(t)}_{i+1}(z) \rceil_2 + \lceil \widetilde{minPivotDist}^{(t)}_{i+1}(v) \rceil_2$.
\end{claim}
\begin{proof}
We show that the path $\langle \ap_{i+1}(z), \hap_{i+1}(z), \hap_{i+1}(v) \rangle$ exists in $H_{i+1}$ and is of small weight. Note that it is possible for some of these vertices to be the same, e.g.~$\ap_{i+1}(z)=\hap_{i+1}(z)$, but this would be a simpler case, as nothing needs to be show for the corresponding edge. Otherwise, we show the existence of each edge on this path one-by-one:
\begin{itemize}
    \item $(\ap_{i+1}(z), \hap_{i+1}(z))$: note that since $z \in V(H_i)$, we have that $\hap^{(t)}_{i+1}(z) = \ap^{(t'')}_{i+1}(z)$ for some $t''$ by the definition of last improving pivots (see \Cref{lastimprPivot}). But note that by \Cref{item:glueEdgesBetweenAPs}, we thus have an edge $(\ap^{(t'')}_{i+1}(z), \ap^{(t)}_{i+1}(z))$ of weight $8 \cdot \lceil \widetilde{minPivotDist}^{(t'')}_{i+1}(z) \rceil_{\roundConst}$. Using that $\hap^{(t)}_{i+1}(z)$ is updated by $\ap^{(t)}_{i+1}(z)$ whenever $\widetilde{minPivotDist}_{i+1}(z)$ improves by at most a factor of two, we can further upper bound the edge weight by $16 \cdot \lceil \widetilde{minPivotDist}^{(t)}_{i+1}(z) \rceil_{\roundConst}$.
    \item $(\hap_{i+1}(z), \hap_{i+1}(v))$: by \Cref{item:glueEdgesBetweenBaseConnectors} this edge is in $E_{i+1}^{base}$ with weight $\lceil \widetilde{minPivotDist}^{(t)}_{i+1}(z) \rceil_2 + \lceil \widetilde{minPivotDist}^{(t')}_i(v) \rceil_2 + \lceil \widetilde{minPivotDist}^{(t)}_{i+1}(v) \rceil_2$.
\end{itemize}
The distance then follows by summing over the upper bounds on the edge weights.
\end{proof}

Straight-forward addition of the upper bounds on the path segments of the  exposed path $\ap_{i+1}(x_s) \leadsto \hap_{i+1}(\overline{x_s}) \leadsto \hap_{i+1}(\overline{y_s}) \leadsto \ap_{i+1}(y_s)$ thus establishes that
\begin{align*}
    dist_{H_{i+1}}(\ap_{i+1}(x_s), \ap_{i+1}(y_s)) \leq &\sum_{z \in \{x_s, y_s, \overline{x_s}, \overline{y_s}\}} 17 \cdot \lceil \widetilde{minPivotDist}_{i+1}(z)\rceil_{\roundConst} \\ &+ \lceil w_{H_j}(\overline{x_s}, \overline{y_s}) \rceil_{\roundConst}  +\sum_{z \in \{ \overline{x_s},  \overline{y_s} \}} \lceil \widetilde{minPivotDist}^{(t')}_i(z) \rceil_2.
\end{align*}
To simplify this expression, we note that 
\[
w_{H_i}(x_s, y_s) =  \lceil \widetilde{minPivotDist}^{(t')}_i(\overline{x_s}) \rceil_2 + w_{H_j}(\overline{x_s}, \overline{y_s}) + \lceil \widetilde{minPivotDist}^{(t')}_i(\overline{y_s}) \rceil_2.\]
For $j < i$, this follows from \Cref{projected_edges} while for $j = i$, we have by \Cref{prop:apprxDist} that $x_s = \overline{x_s}$, $y_s = \overline{y_s}$ and $\widetilde{minPivotDist}_i^{(t')}(\overline{x}) = \widetilde{minPivotDist}_i^{(t')}(\overline{x}) = 0$.

Further, for $z \in \{x_s, y_s\}$, using first \Cref{prop:apprxDist} and then \Cref{fact:pivotIsRoughlyMinPivot} yields
\begin{align*}
\widetilde{minPivotDist}_{i+1}(\overline{z}) &= \sum_{j \leq j' \leq i} \widetilde{pivotDist}_{j'+1}(\ap_{j'}(\overline{z})) \leq 4 \cdot \sum_{j \leq j' \leq i} dist_{H_{j'}}(\ap_{j'}(\overline{z}), V(H_{{j'}+1})) \\
&\leq 4 \cdot \left(\lceil \widetilde{minPivotDist}^{(t')}_i(\overline{z}) \rceil_{\roundConst} + dist_{H_i}(z, \ap_{i+1}(z))\right)
\end{align*}
and therefore
\begin{align*}\widetilde{minPivotDist}_{i+1}(\overline{x_s}) &+ \widetilde{minPivotDist}_{i+1}(\overline{y_s}) \\ 
&\leq 4 \left( w_{H_i}(x_s, y_s) + \sum_{z \in \{x_s, y_s\}} \widetilde{minPivotDist}_{i+1}(z)\right).
\end{align*}
where we use \Cref{fact:pivotIsRoughlyMinPivot} in the last inequality. Our claim now follows by combining these insights.
\end{proof}

\paragraph{Analyzing distances in $H_i$.} For $s = 0, 1, \ldots, \ell-2$, by choice of $x_{s+1}$ and $y_{s+1}$, we have $\widetilde{minPivotDist}_{i+1}(y_{s}) \leq \widetilde{pivotDist}_{i+1}(y_{s}) \leq 4 \cdot dist_{H_i}(y_s, y_{s+1})$ and therefore, using \Cref{prop:apAreCloseToPivots}, also
\begin{align} 
\widetilde{minPivotDist}_{i+1}(x_{s+1}) &\leq \widetilde{pivotDist}_{i+1}(x_{s+1}) \leq 4 \cdot dist_{H_i}(x_{s+1}, V(H_{i+1})) \label{eq:minPivotX}\\
&\leq 4\left(dist_{H_i}(x_{s+1}, y_s) + dist_{H_i}(y_s, \ap_{i+1}(y_s))\right) \nonumber\\
&\leq 4\left(dist_{H_i}(x_{s+1}, y_s) + dist_{H_i}(y_s, \widetilde{pivotDist}_{i+1}(y_s))\right) \nonumber\\
&\leq 16\left(dist_{H_i}(x_{s+1}, y_s) + dist_{H_i}(y_s, y_{s+1})\right) = 16 \cdot dist_{H_i}(x_{s+1}, y_{s+1}).\nonumber
\end{align}
In the last line we used the definition of $x_{s+1}$ and the fact that $y_{s+1}$ is further on $\pi$ from $y_s$. In the last equality we have used the fact that $\pi$ is the shortest path in $H_i$.

Using again \Cref{prop:apAreCloseToPivots}, we further obtain via the triangle inequality
\begin{align}
\widetilde{minPivotDist}_{i+1}(y_{s+1}) &\leq \widetilde{pivotDist}_{i+1}(y_{s+1}) \leq  4 \cdot dist_{H_i}(y_{s+1}, V(H_{i+1})) \label{eq:minPivotY}\\
&\leq 4 \cdot \left( dist_{H_i}(y_{s+1}, y_s) + dist_{H_i}(y_{s}, V(H_{i+1}))\right) \nonumber\\
&\leq 4 \cdot \left( dist_{H_i}(y_{s+1}, y_s) + \widetilde{pivotDist}_{i+1}(y_{s})\right) 
\leq 20 \cdot dist_{H_i}(y_{s+1}, y_s). \nonumber
\end{align}

\paragraph{Analyzing distances in $H_{i+1}$.} Let us recall our analysis: we segmented the path $\pi$ using the sequence $y_0, x_1, y_1, \ldots, y_{\ell-1}, x_{\ell}$ into segments in $H_i$ that we then lift to $H_{i+1}$. The total weight of segments $y_s \leadsto x_{s+1}$ lifted to $H_{i+1}$ is  
\begin{align*}
    \sum_{s=0}^{\ell-1} w_{H_{i+1}}(&\ap_{i+1}(y_s), \ap_{i+1}(x_{s+1}))
    = 8 \sum_{s=0}^{\ell-1} \lceil \widetilde{pivotDist}_{i+1}(y_s) \rceil_{\roundConst} = 8 \sum_{s=1}^{\ell-1} \lceil \widetilde{pivotDist}_{i+1}(y_s) \rceil_{\roundConst} 
\end{align*}
where we use in the last equality that $y_0 = u \in V(H_{i+1})$ which implies that $\widetilde{pivotDist}_{i+1}(y_0) = 0$ by \Cref{prop:apAreCloseToPivots}. It remains to use \Cref{eq:minPivotY} to deduce that 
\[
\sum_{s=0}^{\ell-1} w_{H_{i+1}}(\ap_{i+1}(y_s), \ap_{i+1}(x_{s+1})) \leq 160 \cdot dist_{H_i}(u,v).
\]

For the second type of segments, we can bound the weight of these segments using \Cref{lma:secondtypeEdgesLift} to obtain
\begin{align*}
    \sum_{s = 1}^{\ell -1} dist_{H_{i+1}}&(\ap_{i+1}(x_s),\ap_{i+1}(y_s)) \\
    &\leq \sum_{s = 1}^{\ell -1} \left(69 \cdot w_{H_i}(x_s, y_s) + \sum_{z \in \{x_s, y_s\}} 85 \cdot \lceil \widetilde{minPivotDist}_{i+1}(z)\rceil_{\roundConst}\right)\\
    & \leq 69 \cdot dist_{H_i}(u,v) + 85 \cdot (40 \cdot dist_{H_i}(u,v))\\
    &\leq 3469 \cdot dist_{H_i}(u,v).
\end{align*}
Combining the weights of these two segment types yields $dist_{H_{i+1}}(u,v) \leq 3629 \cdot dist_{H_i}(u,v)$, as desired.

\subsection{An Algorithm for Maintaining the Hierarchy}

The main result of this section is summarized in the following theorem.

\begin{restatable}{theorem}{mainAlgoResult}\label{thm:algoToMaintainVertexSparsifierHier}
Given an incremental, undirected, weighted graph $G=(V,E,w)$, there is a deterministic algorithm that maintains the hierarchy of vertex sparsifiers $H_1, H_2, \ldots, H_k$ as described in \Cref{def:sparsifier hierarchy} for some $k = \Theta(\log\log n)$. Additionally, the algorithm answers queries given a level $1 \leq i \leq k$, and vertices $u,v \in V(H_i)$ where the query returns a distance estimate $ \widehat{dist}(u,v)$ that satisfies $dist_{H_i}(u,v) \leq \widehat{dist}(u,v)$ and if $u \in B_{H_i}(v, \frac{1}{8}\widetilde{pivotDist}_{i+1}(v))$ or $v \in B_{H_i}(u, \frac{1}{8}\widetilde{pivotDist}_{i+1}(u))$ where we define $\widetilde{pivotDist}_{k+1}(v) = \infty$, it is further guaranteed that $\widehat{dist}(u,v) \leq 2 \cdot dist_{H_i}(u,v)$. 

The algorithm maintains the vertex sparsifiers in total time $O(kn \log^6 nW + km)$ and answers every query in worst-case $O(1)$ time.
\end{restatable}

For the rest of this section, we assume w.l.o.g. that $G$ is initially connected, has diameter at most $n^2 W$\footnote{We use $n^2 W$ factor, so that we can use the connectivity assumption w.l.o.g.~as otherwise we would add a super source with weights $nW$.} and all edge weights in $[2, W]$. 

Let us now start by giving an algorithm to maintain the approximate pivots $\ap_{i+1}(v)$ for each $v \in V(H_i)$. This also allows us to determine the vertex sets of each graph $H_{i+1}$. Once this algorithm is set-up, we give an algorithm to maintain historic approximate pivots $\hap_{i+1}(v)$ for each $v \in V$. Finally, we discuss how to maintain the edges in the graph hierarchy. We note that for technical reasons, all algorithms work on the graphs $\wH_1, \wH_2, \ldots, \wH_k$ where $\wH_i$ is the graph $H_i$ with all edges rounded up to the nearest power of $\roundConst$.

\begin{algorithm}
\DontPrintSemicolon
\For(\label{lne:forEachInMainUpdate}){$i = 1, \ldots, k-1$}{

    \While(\label{lne:whileLoopPivots}){$\exists v \in V(\wH_i)$ such that $|B_{\wH_i}(v, \frac{1}{4}\widetilde{pivotDist}_{i+1}(v))| \geq \hat{b}_i$}{
     
        Let $B_v$ be a set of size $\hat{b}_i$ such that $B_v \subseteq B_{\wH_i}(v, \frac{1}{4}\widetilde{pivotDist}_{i+1}(v))$.\label{lne:computeBClosest}\\
     
        \If(\label{lne:ifCaseMainUpdate}){$\exists u \in B_v$ with $\widetilde{pivotDist}_{i+1}(u) < \frac{1}{2} \widetilde{pivotDist}_{i+1}(v)$}{
            $p_{i+1}(v) \gets p_{i+1}(u)$; $\widetilde{pivotDist}_{i+1}(v) \gets dist_{\wH_i}(v,u) + \widetilde{pivotDist}_{i+1}(u)$.\label{lne:ifCasePivotUpdate}
        }\Else{
            Add $v$ to $\wH_{i+1}$.\label{lne:addPivot}\\
            \ForEach(\label{lne:forEachLoopPivot}){$u \in B_{v}$}{
                $p_{i+1}(u) \gets v$; $\widetilde{pivotDist}_{i+1}(u) \gets dist_{\wH_i}(u,v)$.\label{lne:elseCasePivotUpdate}
            }
        }
    }
}
\caption{$\textsc{UpdateApproxPivots}()$}
\label{alg:updateAp}
\end{algorithm}

\paragraph{Parameters.} Throughout the section, we use parameters $b_i = 2^{(6/5)^i}$ for any $i \geq 1$, and let $k$ be the smallest index such that $\prod_{i \leq k} b_{i} > n$. It is straight-forward to calculate that $k = \Theta(\log\log n)$. For convenience, we define $\hat{b}_i = b_i \cdot  (\log_{4/3}(n^2 W) + 1)$ for each $i$. 

\paragraph{Maintaining Approximate Pivots (Pseudo-code).} We initialize for each $u \in V$, the pivot $p_2(u) = \bot$ and let $\widetilde{pivotDist}_{2}(u) = n^2 W$. We initialize $H_1$ to $G$ and $H_2, H_3, \ldots, H_k$ to empty graphs (and initialize $\wH_1, \wH_2, \ldots, \wH_k$ to empty graphs). Throughout the algorithm, whenever a new vertex $v$ is added to vertex set $V(H_i)$ (and thus also to $V(\wH_i)$), we again initialize its pivot $p_{i+1}(u) = \bot$ and let $\widetilde{pivotDist}_{i+1}(u) = n^2 W$.

After this initialization and after each update to $G$, we invoke $\textsc{UpdateApproxPivots}()$ given in \Cref{alg:updateAp}. The goal of the algorithm is two-fold:
\begin{itemize}
    \item (Ball Size Constraint) Intuitively, we want the ball $B_{H_i}(v, d_{H_i}(v, V(H_{i+1})))$ to contain at most $\hat{b}_i$ vertices for each $v \in V(H_i)$, so that the local computation can be done more efficiently. If this constraint is violated, we need to take action and make a new vertex in this ball a pivot of $v$ so that the ball shrinks in size. Since we work with approximate pivots, however, we have to relax this constraint. To counter this, we search even more aggressively for an approximate pivot that enforces this constraint on the ball $B_{\wH_i}(v, \frac{1}{4}\widetilde{pivotDist}_{i+1}(v))$. Since we have $d_{H_i}(v, \ap_{i+1}(v)) \leq \widetilde{pivotDist}_{i+1}(v) \leq 4 \cdot dist_{H_i}(v, V(H_{i+1}))$ by \Cref{prop:apAreCloseToPivots} of \Cref{def:sparsifier hierarchy}, we thus have that $|B_{\wH_i}(v, \frac{1}{4}\widetilde{pivotDist}_{i+1}(v))| \leq |B_{H_i}(v, d_{H_i}(v, V(H_{i+1})))| \leq \hat{b}_i$. 
    \item (Graph Size Constraint) on the other hand, we also want $V(H_{i+1})$ to be significantly smaller than $V(H_i)$ (roughly by a factor $\hat{b}_i$). Since $V(H_{i+1})$ is in the image of $p_{i+1}$, we have to make each approximate pivot $\ap_{i+1}(v)$ a vertex of $V(H_{i+1})$. At an extreme, while making each vertex in $V(H_i)$ its own approximate pivot is a viable choice of $\ap_{i+1}$ with respect to ball sizes, it is a poor choice when considering the size of $V(H_{i+1})$. Therefore, we use a dynamic covering technique that allows us to bound the number of vertices that are in $V(H_{i+1})$ (i.e. at any point in the image of $p_{i+1}$) by a much smaller factor.
\end{itemize}

Our algorithm optimizes these two constraints using a simple rule: whenever a new pivot is required due to a ball size constraint being violated, such a vertex $v$ in need, first asks other vertices that are close to it if their approximate pivot is a good fit. Otherwise, $v$ becomes a pivot itself, but also the new pivot of these close vertices. We note that the algorithm also already shows how to maintain $\widetilde{pivotDist}_{i+1}(v)$ for all $v \in V(H_i)$ for all $i$. Below, we establish our claim on the graph size.

\begin{claim}\label{clm:decreaseInApproxDist}
Whenever the approximate pivot $\ap_{i+1}(v)$ of a vertex $v \in V(H_i)$ is changed, the value  $\widetilde{pivotDist}_{i+1}(v)$ decreases to a $\frac{3}{4}$-fraction of the original value. 
\end{claim}
\begin{proof}
It is not hard to see from \Cref{alg:updateAp} that for each $v \in V(H_i)$, $\widetilde{pivotDist}_{i+1}(v)$ is monotonically decreasing over time. Further, when the pivot of a vertex $v$ is changed in \Cref{lne:ifCasePivotUpdate}, we have that $\widetilde{pivotDist}^{NEW}_{i+1}(v) = dist_{H_i}(v,u) + \widetilde{pivotDist}_{i+1}(u) \leq \frac{3}{4} \widetilde{pivotDist}^{OLD}_{i+1}(v)$ by the if-condition.

If the pivot of a vertex $u$ is changed in the for-each loop in \Cref{lne:forEachLoopPivot}, we have that
$\widetilde{pivotDist}^{NEW}_{i+1}(u) = dist_{H_i}(u,v) \leq \frac{1}{4}\widetilde{pivotDist}^{OLD}_{i+1}(v) \leq \frac{1}{2}\widetilde{pivotDist}^{OLD}_{i+1}(u)$ where the last inequality stems from \Cref{lne:computeBClosest} and the fact that the if condition in \Cref{lne:ifCaseMainUpdate} was not satisfied.
\end{proof}
\begin{corollary}\label{cor:numberOfVerticesDecreasesALot}
At any stage, we have $|V(H_{i+1})| \leq \frac{|V(H_i)| (\log_{4/3}(n^2 W) + 1)}{\hat{b}_i} = \frac{|V(H_i)|}{b_i}$.
\end{corollary}
\begin{proof}
By \Cref{clm:decreaseInApproxDist}, we have that every vertex can change its approximate pivot at most $\log_{4/3}(n^2 W) + 1$ times. But note that whenever a new vertex $v$ is added to the set $V(H_{i+1})$ (which happens only in \Cref{lne:addPivot}), we change the approximate pivots of all vertices in the current ball $B_{H_i}(v, \frac{1}{4}\widetilde{pivotDist}_{i+1}(v))$ which has size at least $\hat{b}_i$ by \Cref{lne:computeBClosest}. 
\end{proof}

For convenience, we define $n_i$ for each level $i$ the final size of set $V(H_i)$ for the rest of this section. Thus, the corollary above can be restated as $n_{i+1} \leq n_i / b_i$. 

\paragraph{Maintaining Approximate Pivots (Implementation).} Next, we propose an efficient implementation of \Cref{alg:updateAp} given that there is a procedure that updates $\wH_i$ based on the updated approximate pivots/ approximate pivot distances. 

In our implementation, we use the following crucial primitive $\textsc{TruncDijkstra}(v, H_i, r)$: for any vertex $v \in V(H_i)$ and radius $r$, we run Dijkstra's algorithm from $v$ in $\wH_i$ where we stop relaxing vertices that are at distance greater than $r$, or abort after having found $\hat{b}_i$ such vertices. But note that by definition of the adjacency list $\textsc{Adj}_{\wH_i, u}$ of each vertex $u$ in graph $H_i$, we can use exclusively the edges in $\textsc{Adj}_{\wH_i, u}[1, \hat{b}_i]$ for each $u$ that we relax and therefore implement the algorithm efficiently in $O(\hat{b}_i^2 + \hat{b}_i \log \hat{b}_i)$ time (recall from \Cref{sec:preliminaries} that these lists are ordered by weight and time of arrival). Further, we can store with $v$ a deterministic dictionary $\mathcal{D}_i(v)$ (see \cite{HAGERUP200169}) that allows us to check for each vertex $u$ if $u$ is one of the $\hat{b}_i$ vertices relaxed by Dijkstra's algorithm, and if so, we can return the distance $dist_{\wH_i}(u,v)$. The construction time of the dictionary is subsumed by the bound $O(\hat{b}_i^2 \log \hat{b}_i)$ and its query time is worst-case constant. 

Equipped with this primitive, let us give the entire algorithm. Throughout each stage, we maintain a list of \emph{unvisited} vertices $\textsc{Unvisited} \subseteq V(\wH_i)$ that corresponds to vertices where we cannot currently ensure that the while-loop condition in \Cref{lne:whileLoopPivots} holds. 

At the initial stage, we have $\textsc{Unvisited}$ equal to $V(\wH_i)$, i.e. the initial set of vertices of $\wH_i$. At the beginning of any subsequent stage, $\textsc{Unvisited}$ consists of the vertices $u \in V(H_i)$ for which there exists a vertex $v \in \mathcal{D}_u$ where $\textsc{Adj}_{\wH_i,v}[1, \hat{b}_i]$ was updated since the last stage. 

Then, at any stage, once $\textsc{Unvisited}$ is initialized as described above, we do the following: while there exists a vertex $v \in \textsc{Unvisited}$, we run $\textsc{TruncDijkstra}(v, \wH_i, \frac{1}{4}\widetilde{pivotDist}_{i+1}(v))$. If the primitive explores less than $\hat{b}_i$ vertices, we store the dictionary $\mathcal{D}_i(v)$ and remove $v$ from $\textsc{Unvisited}$. Otherwise, i.e. if the primitive explores $\hat{b}_i$ vertices for $v$, then we enter the while-loop. We can obtain $B_v$ as described in \Cref{lne:computeBClosest} from the primitive (by scanning the dictionary) and it is not hard to see that the rest of the while-loop iteration can be implemented in time $O(\hat{b}_i)$.

We prove next that this implementation is efficient.

\begin{lemma}\label{lma:approxPivotsAreReallyPowerful}
Given a procedure that updates each $\wH_{i}$ just before the $i$-th iteration of \Cref{alg:updateAp} in such a way that for each $u$ keeps the adjacency list of $u$ in $\wH_i$ ordered by weight and time of arrival, the total update time of all invocations of \Cref{alg:updateAp} (excluding the time required by the procedure updating each $\wH_i$) can be bounded by $O\left(\sum_i n_i b_i^4 \log^6 nW + \Delta \right)$ where $\Delta$ is the total number of edges and vertices in the final versions of $H_1, H_2, \ldots, H_k$. 

Additionally we can query given a level $1 \leq i \leq k$, and vertices $u,v \in V(H_i)$ where the query returns a distance estimate $ \widehat{dist}(u,v)$ that satisfies $dist_{H_i}(u,v) \leq \widehat{dist}(u,v)$ and  if $u \in B_{H_i}(v, \frac{1}{8}\widetilde{pivotDist}_{i+1}(v))$ or $v \in B_{H_i}(u, \frac{1}{8}\widetilde{pivotDist}_{i+1}(u))$ where we define $\widetilde{pivotDist}_{k+1}(v) = \infty$, it is further guaranteed that $\widehat{dist}(u,v) \leq 2 \cdot dist_{H_i}(u,v)$. 
\end{lemma}
\begin{proof}
We first argue for correctness. Note that each vertex $v$ that is on the list $\textsc{Unvisited}$ and then removed at the end of the stage cannot satisfy the condition of the while-loop in \Cref{lne:whileLoopPivots}. This follows from the fact that a vertex $v$ is only removed from the list $\textsc{Unvisited}$ when the primitive  $\textsc{TruncDijkstra}(v, \wH_i, \frac{1}{4}\widetilde{pivotDist}_{i+1}(v))$ certifies that it cannot satisfy the while-loop condition. Further, we have that $\widetilde{pivotDist}_{i+1}(v)$ is monotonically decreasing over time by \Cref{clm:decreaseInApproxDist} and therefore the condition remains true. It remains to argue that we can initialize $\textsc{Unvisited}$ at a stage after the initial stage to only consist of vertices $u$ that had no vertex $v \in \mathcal{D}_u$ with updated $\textsc{Adj}_{\wH_i,v}[1, \hat{b}_i]$. But this implies that the primitive $\textsc{TruncDijkstra}(u, \wH_i, \frac{1}{4}\widetilde{pivotDist}_{i+1}(u))$ would relax the same vertices as at the previous stage. It is thus straight-forward to prove by induction that for $v$ the while-loop condition does not hold.

We observe first that for each vertex $v$, its adjacency list $\textsc{Adj}_{\wH_i,v}[1, \hat{b}_i]$ can be updated at most $\hat{b}_i \log_2(nW)$ times over the entire course of the algorithm by the way the ordering of edges incident on $v$ is determined and by the fact that $\wH_i$ only allows for edge weights that are powers of $\roundConst$ and is incremental.

For the running time, we first use that for each vertex $v \in V(H_i)$, we have that $\widetilde{pivotDist}_{i+1}(v)$ is decreased at most $O(\log(nW))$ times over the entire course of the algorithm by \Cref{clm:decreaseInApproxDist}. Further, between any two times that $\widetilde{pivotDist}_{i+1}(v)$ is decreased, we claim that the primitive $\textsc{TruncDijkstra}(v, \wH_i, \frac{1}{4}\widetilde{pivotDist}_{i+1}(v))$ is invoked at most $O(\hat{b}_i)$ times. This follows since between these times $\widetilde{pivotDist}_{i+1}(v)$ remains fixed and since $\wH_i$ is an incremental graph, $B_{\wH_i}(v, \frac{1}{4}\widetilde{pivotDist}_{i+1}(v))$ is monotonically increasing. However only until it contains $\hat{b}_i$ or more vertices, as this triggers that the while-loop condition is violated on $v$ again. Until then, each of the at most $\hat{b}_i - 1$ vertices that are in the ball just before it starts violating the while-loop condition can have $\hat{b}_i \log_2(nW)$ updates to their adjacency list $\textsc{Adj}_{\wH_i,u}[1, \hat{b}_i]$ triggering an additional invocation of $\textsc{TruncDijkstra}(v, \wH_i, \frac{1}{4}\widetilde{pivotDist}_{i+1}(v))$. Summarizing our discussion, we can bound the number of times that $\textsc{TruncDijkstra}(v, \wH_i, \frac{1}{4}\widetilde{pivotDist}_{i+1}(v))$ is run for some vertex $v \in V(H_i)$ to be at most $O(\hat{b}^2_i \log^2 nW)$. Using that the time spent by of each such call can be upper bound by $O(\hat{b}_i^2)$, and given that these calls asymptotically subsumes all other operations of the implementation of \Cref{alg:updateAp}, we thus arrive at the runtime stated above.

To prove that we can carry out queries as stated, we require two insights: 1) $dist_{H_i}(u,v) \leq dist_{\wH_i}(u,v) \leq 2dist_{H_i}(u,v)$ which is trivial from the fact that each $\wH_i$ is  derived from $H_i$ by rounding up weights to the nearest power of $2$ and 2) the fact that $V(H_k) = \emptyset$ which follows from \Cref{cor:numberOfVerticesDecreasesALot} which implies $|V(H_k)| \leq n / \prod_{j=1}^k b_j$ and the fact that $k$ is chosen such that $\prod_{j=1}^k b_j > n$. Given these two insights, it is not hard to verify that the dictionaries $\mathcal{D}_i(v)$ that we have stored at the end of each stage enable us to carry out the stated queries where we obtain for a vertex $u$ the exact distance $dist_{\wH_i}(u,v)$ and return it as a distance estimate or if we cannot find an entry in the dictionary we can simply return $\infty$.
\end{proof}

\paragraph{Maintaining Historic Approximate Pivots.} Before we can describe how to maintain the graphs $H_1, H_2, \ldots, H_k$ (and thus $\wH_1, \wH_2, \ldots, \wH_k$), it is straight-forward to see from \Cref{def:sparsifier hierarchy}, that we also need to maintain the last improving pivots $\hap_{i+1}(v)$. Therefore, we need to know the current pivot distances $\widetilde{pivotDist}_{i+1}(v)$ for each $v \in V$ (recall that \Cref{alg:updateAp} maintains these distances only for vertices in $V(H_i)$). More precisely, we need an algorithm that informs us when $\lceil \widetilde{minPivotDist}_{i+1}(v) \rceil_{\roundConst}$ decreases. 

Focusing on a given level $i \leq k-1$, we keep explicit variables $\widetilde{mpd}_{i+1}(v)$ for each $v \in V = V(H_1)$. Our algorithm ensures that we have $\widetilde{mpd}_{i+1}(v) = \lceil \widetilde{minPivotDist}_{i+1}(v) \rceil_{\roundConst}$ by the end of each stage. To achieve this goal, consider the following natural hierarchy forest $F_i$. We let the vertices of $F_i$ correspond to the vertices in $V(H_1), V(H_2), \ldots, V(H_{i+1})$, where we have an edge from each vertex $v \in V(H_j), j \leq i$ to its current approximate pivot, i.e. an edge $(v, \ap_{j+1}(v))$. Using directed edges, it is clear that the vertices in $V(H_{i+1})$ form the roots of the forest $F_i$, the vertices in $V(H_j)$ form the level-$j$ vertices (i.e. at distance $j-1$ from the leaves) and $V(H_1)$ form the leaf vertices.

We further maintain the following weight function $w_{F_i}$ over the edges:
\begin{align*}
    w_{F_i}(e = (v, \ap_{j+1}(v))) = \begin{cases}
    \widetilde{pivotDist}_{j+1}(v) - \widetilde{mpd}_{i+1}(v) & \text{if }$j = 1$\\
    \widetilde{pivotDist}_{j+1}(v) & \text{otherwise}
    \end{cases}
\end{align*}

We can maintain this collection of trees $F_i$ using the dynamic tree data structure introduced below.

\begin{theorem}[see \cite{alstrup2005maintaining}, Theorem 2.7]
Given a directed forest $F=(V,E,w)$ with (possibly negative) edge weights $w$, there is a data structure $\mathcal{D}$ that supports the following operations:
\begin{itemize}
    \item $\textsc{AddEdge}(e, w_e)$ / $\textsc{DeleteEdge}(e)$: adds an edge $e$ with weight $w_e$ (assuming that the tail of $e$ is a root) / deletes an edge $e$ from $F$.
    \item $\textsc{FindRoot}(u)$: Returns the root of a vertex $u$.
    \item $\textsc{ReturnDist}(u)$: Returns the distance from $u$ to its root.
    \item $\textsc{Mark}(u)$ / $\textsc{Unmark}(u)$: Marks / Unmarks a vertex $u$. Initially all vertices are unmarked.
    \item $\textsc{FindNearestMarkedVertex}(u)$: Finds the vertex in $u$'s subtree that is at closest distance and marked (if such a vertex exists).
\end{itemize}
The data structure can be initialized in $O(n \log n)$ time and implement each operation in $O(\log n)$ time.
\end{theorem}

We can now state our algorithm to maintain the correct values $\widetilde{mpd}_{i+1}(v)$ for each $v \in V = V(H_1)$. On initialization, we mark all vertices in $V(H_1)$ and leave all other vertices unmarked. We initialize for each $v \in V$, $\widetilde{mpd}_{i+1}(v) = \lceil n^2 W \rceil_{\roundConst}$.

Next, consider an update to $G$. This potentially results in many changes of edges in $F_i$ and weights due to changes in the distances to approximate pivots $\widetilde{pivotDist}_{j+1}(v)$. Our algorithm starts by forwarding these changes to the data structure $\mathcal{D}$ (for the initial change we assume the data structure is initially empty so the entire forest $F_i$ is encoded in these changes). Weight changes are implemented by first deleting edges and then adding them back into the forest with their new weight.

Then, for each tree $T \in F_i$ that underwent some change at the current stage, we find its root $r \in V(H_{i+1})$ and query for $u = \textsc{FindNearestMarkedVertex}(r)$, the nearest leaf of $r$. If the distance from $u$ to $r$ is non-negative, the algorithm moves on to the next tree. Otherwise, it sets $\widetilde{mpd}_{i+1}(u) = \lceil \widetilde{minPivotDist}(u) \rceil_2$ which can be extracted from the distance from $u$ to $r$. Then, the algorithm repeats this the procedure for the current tree. 

\begin{claim}\label{clm:correctnessOfHistoricPivots}
At the end of each stage, we have for each $v \in V$, \[\widetilde{mpd}_{i+1}(v) = \lceil \widetilde{minPivotDist}_{i+1}(v) \rceil_{\roundConst}.\]
\end{claim}
\begin{proof}
Whenever $\widetilde{mpd}_{i+1}(v)$ is re-set, it is set to the current value $ \lceil \widetilde{minPivotDist}(u) \rceil_2$. Since $\lceil \widetilde{minPivotDist}(u) \rceil_2$ is monotonically decreasing over time, we thus have $\widetilde{mpd}_{i+1}(v) \geq \lceil \widetilde{minPivotDist}(u) \rceil_2$.

On the other hand, by our algorithm, we ensure that at the end of every stage, for each tree $T \in F_i$, that the nearest leaf to its root is at a non-negative distance. Thus, it is not hard to see that for each $v \in V(H_1)$, we have $\sum_{j \leq i} \widetilde{pivotDist}_{j+1}(v) \geq \widetilde{mpd}_{i+1}(v)$. But since 
$\widetilde{minPivotDist}(u)$ is exactly the minimum over all such sums $\sum_{j \leq i} \widetilde{pivotDist}_{j+1}(v)$ at previous and the current stages, we have that $\widetilde{mpd}_{i+1}(v) \leq \widetilde{minPivotDist}(u)$. Rounding up both quantities to the nearest power of two further preserves this inequality.
\end{proof}

\begin{claim}\label{clm:runtimeHistoricPivots}
Given an algorithm \Cref{alg:updateAp} to update for each $1 \leq i < k$, $v \in V(H_i)$ the quantities $\widetilde{pivotDist}_{i+1}(v)$ and $\ap_{i+1}(v)$, we can maintain for each $1 \leq i < k$ and $v \in V$, the last improving pivots $\hap_{i+1}(v)$ and $\lceil \widetilde{minPivotDist}_{i+1}(v) \rceil_{\roundConst}$ in additional time $O(kn \log nW \log n)$.
\end{claim}
\begin{proof}
Fixing a level $1 \leq i < k$, we have that each vertex in $V(H_i)$ has $\widetilde{pivotDist}_{i+1}(v)$ and $\ap_{i+1}(v)$ updated at most $O(\log nW)$ times. Since each update can be handled by the algorithm described above in time $O(\log n)$, our algorithm spends at most $O(|V(H_i)| \log nW \log n)$ time on handling updates and the resulting query on the updated tree. Further, each time a query returns a negative value, we at least half the value $\widetilde{mpd}_{i+1}(u)$ for some $u \in V$. Thus, the number of such queries is bound by $O(n \log nW)$ and each query and subsequent update of $\widetilde{mpd}_{i+1}(u)$ can again be implemented in time $O(\log n)$. Since we have $k-1$ levels where we maintain our data structure, the bound follows.
\end{proof}

\paragraph{Putting it all together.} Finally, we give the algorithm to maintain $H_1, H_2, \ldots, H_k$ (and thus to maintain $\wH_1, \wH_2, \ldots, \wH_k$). Recall that initially these graphs are equal to the empty graphs. Then, $H_i$ (and $\wH_i$) is updated whenever \Cref{alg:updateAp} is invoked, and to be precise, is updated just before the $i$-th iteration of the for-loop in \Cref{alg:updateAp}. For $i = 1$, the update is simple as $H_1$ is just equal to $G$ at any stage. For $i + 1 \geq 1$, we have that $H_i$ is already updated for the current stage, and for any $v \in V(H_i)$, we have $\ap_{i+1}(v), \widetilde{pivotDist}_{i+1}(v), B_{H_i}(v, \frac{1}{4}\widetilde{pivotDist}_{i+1}(v))$ and $\mathcal{D}_v$, and for any $v \in V$, we have $\hap_{i+1}(v)$ and $\lceil \widetilde{minPivotDist}_{i+1}(v) \rceil_{\roundConst}$ in their updated version (i.e. these values do not change for the rest of the stage). This can be seen easily from \Cref{alg:updateAp} and our description of the algorithm to maintain historic approximate pivots. 

Given this updated information, it is straight-forward to generate all edges that are missing from $H_{i+1}$ in constant additional time per edge added to $H_{i+1}$. To obtain a bound on the runtime, it thus suffices to bound the number of edges in $H_{i+1}$.

\begin{lemma}\label{lma:edgesUpperBound}
Throughout the algorithm, we have for any $1 \leq i < k$, that $|E(H_{i+1})| = O(m + kn \log^5 nW + \sum_{j \leq i} |V(H_j)|b_j \log^4 nW)$.
\end{lemma}
\begin{proof}
We first prove the claim that for any such $i$, we have $|E^{base}_{i+1}| = O(n \log^3 nW + |V(H_i)|\hat{b}_i \log nW)$. We proceed by a case analysis for each edge type that is generated. Following \Cref{def:sparsifier hierarchy} we have
\begin{itemize}
    \item For edges generated from \Cref{iball_edges}, we have that there is at most one edge generated for each vertex $v$ in $B_{H_i}(u, \frac{1}{4} \cdot \widetilde{pivotDist}_{i+1}(u))$ for any vertex $u \in V(H_i)$. Since $\widetilde{pivotDist}_{i+1}(u)$ is updated at most $O(\log nW)$ times and since it is ensured that $B_{H_i}(u, \frac{1}{4} \cdot \widetilde{pivotDist}_{i+1}(u))$ is of size less than $\hat{b}_i$ when edges are generated, we have that there are at most $O(|V(H_i)|\hat{b}_i \log nW)$ such edges. 
     
     \item For edges generated from \Cref{item:glueEdgesBetweenAPs}, we have that such edges are only generated for a vertex $v \in V(H_i)$ whenever its pivot $\ap_{i+1}(v)$ is updated, and if so an edge to every former pivot of $v$ is added. But since we have from  \Cref{clm:decreaseInApproxDist} that there are at most $O(\log nW)$ such pivots throughout the algorithm, the total number of such edges can be bound by $O(|V(H_i)|\log^2 nW)$.
     
    \item For edges generated from \Cref{item:glueEdgesBetweenBaseConnectors}, we have to generate new edges for a vertex $v \in V$ only when $\hap_{i+1}(v)$ is updated or if for an existing edge $(\hap^{(t)}_{i+1}(x), \hap_{i+1}^{(t)}(v))$ where $x = \hap^{(t')}_{i}(v)$ for some $t' \leq t$, the weight has to change because one of the quantities $\lceil \widetilde{minPivotDist}^{(t)}_{i+1}(x) \rceil_2, \lceil \widetilde{minPivotDist}^{(t')}_i(v) \rceil_2, \lceil \widetilde{minPivotDist}^{(t)}_{i+1}(v) \rceil_2$ was changed. By our previous line of argumentation, we have that there are at most $O(\log nW)$ historic pivots $\hap_i(v)$ and $\hap_{i+1}(v)$ and that $\lceil \widetilde{minPivotDist}^{(t)}_{i+1}(x) \rceil_2, \lceil \widetilde{minPivotDist}^{(t)}_{i+1}(v) \rceil_2$ can change at most $O(\log nW)$ times for each pair $x,v$ for which we generate an edge. We conclude that there are at most $O(n \log^3 nW)$ such edges. 
\end{itemize}

Having established the upper bound on $E_{j}^{base}$ for all $1 \leq j < k$, and using that $E_1^{base}$ consists of the edges in $G$, we can thus bound the number of edges in $E_{i+1}^{proj}$. These edges are generated from \Cref{projected_edges}, where we have an edge $(\hap_{i+1}(x), \hap_{i+1}(y))$ in $E_{i+1}^{proj}$ for each edge $(x,y) \in E^{base}_j$ for any $j \leq i$ for any historic pivots of $x$ and $y$ and the edge weight is equal to $ \lceil \widetilde{minPivotDist}_{i+1}(x)\rceil_{\roundConst} + \lceil w_{H_j}(e) \rceil_{\roundConst} +  \lceil \widetilde{minPivotDist}_{i+1}(y)\rceil_{\roundConst}$. From our previous discussion, we have that for each such edge in $E^{base}_j$ for any $j \leq i$, we can have at most $O(\log^2 nW)$ versions in $E_{i+1}^{proj}$. Thus, the total number of such edges is $O(m + kn \log^5 nW + \sum_{j \leq i} |V(H_j)|\hat{b}_j \log^3 nW)$, as desired.
\end{proof}

Using this upper bound on the number of edges, we can prove the main result of the section, \Cref{thm:algoToMaintainVertexSparsifierHier}.

\mainAlgoResult*

\begin{proof}
We have correctness of the algorithm by \Cref{lma:approxPivotsAreReallyPowerful}, \Cref{clm:correctnessOfHistoricPivots} and our previous discussion. It thus remains only to bound the runtime. From \Cref{lma:approxPivotsAreReallyPowerful}, \Cref{clm:runtimeHistoricPivots} and \Cref{lma:edgesUpperBound}, we can upper bound the total update time of the algorithm by 
\[O\left(\sum_{i=1}^k n_i b_i^4 \log^6 nW + mk \log(m) \right)\] 
(recall $n_1 = \Omega(n)$). Further note that for $i < 10$, we have $b_i \leq 2^{(6/5)^{10}} = O(1)$ and therefore $n_i  b_i^4 = O(n)$. For $i \geq 10$, we have that 
\[
n_i \leq n / \prod_{j = 1}^{i-1} b_{j} = n / 2^{\sum_{j = 1}^{i-1} (6/5)^j} = n / 2^{(6/5)^i \sum_{j = 1}^{i-1} (5/6)^j} \leq n/ b_i^4
\]
where we use the formula for geometric sums to obtain $\sum_{j = 1}^{i-1} (5/6)^{j} \geq (1-(5/6)^{-10})/(1-5/6) - 1 \geq 4$. This allows to bound the total update time as stated above, and it remains to observe that one can implement the data structure to maintain $\textsc{Adj}_{v, H_i}$ for each $i$ in time $O(m \log m)$ by using binary search when adding new edges in to the adjacency list of $v$.
\end{proof}

\subsection{The Query Algorithm}

Finally, we can give a query algorithm that is almost identical to the query algorithm in \cite{andoni2020parallel}. Here we define for convenience the function $\ap_1 : V \mapsto V$ to be the identity function and recall that $\widetilde{pivotDist}_{k+1}(x) = \infty$ for all $x \in V(H_k)$.

\begin{algorithm}
\DontPrintSemicolon
$i \gets 1$.\\
\While(\label{lne:whileLoopConditionQuery}){the distance estimate $\widehat{dist}(\ap_i(u),\ap_i(v))$ from \Cref{thm:algoToMaintainVertexSparsifierHier} exceeds $\max\{\frac{1}{4}\widetilde{pivotDist}_{i+1}(\ap_i(v)), \frac{1}{4}\widetilde{pivotDist}_{i+1}(\ap_i(u))\}$}{
    $\widetilde{d}_i \gets \widetilde{pivotDist}_{i+1}(\ap_i(u)) + \widetilde{pivotDist}_{i+1}(\ap_i(v))$.\label{lne:addPivotPaths}\\
    $i \gets i + 1$.
}
$\widetilde{d}_i \gets \widehat{dist}(\ap_i(u),\ap_i(v))$.\label{lne:addFinalPathDist}\\
\Return $\widetilde{dist}(u,v) = \sum_{j \leq i} \widetilde{d}_j$
\caption{$\textsc{QueryDist}(u,v)$}
\label{alg:query}
\end{algorithm}

\begin{lemma}
The algorithm $\textsc{QueryDist}(u,v)$ returns a distance estimate $\widetilde{dist}(u,v)$ such that 
\[
    dist_G(u,v) \leq \widetilde{dist}(u,v) \leq \tilde{O}(1) \cdot dist_G(u,v).
\]
The algorithm runs in worst-case time $O(\log\log n)$.
\end{lemma}

\paragraph{Runtime Analysis.} To bound the runtime, we first observe that if we are in the $j$-th iteration of the while-loop, once we evaluated $p_j(u)$ and $p_j(v)$, we can implement the while-loop in $O(1)$ time. While $p_j(u)$ and $p_j(v)$ are nested functions of depth $j-1$ which would naively take time $O(j)$ to evaluate, we use that in the $j-1$-th iteration (for $j > 1$, otherwise $O(j)$ is constant), we already evaluate $p_{j-1}(u)$ and $p_{j-1}(v)$, and by keeping them stored in a cache, we can compute $p_j(u) = p_j(p_{j-1}(u))$ and $p_j(v) = p_j(p_{j-1}(v))$ in constant time respectively. Thus, each iteration of the while loop takes constant time.

Letting $i$ refer to the final value of the variable, we claim that $i \leq k$. Given this claim, it is not hard to see that from \Cref{alg:query} and \Cref{thm:algoToMaintainVertexSparsifierHier}, that the total time spend can be bound by $O(k) = O(\log\log n)$.

To prove that $i \leq k$, we observe that the variable is initialized to $1$ and then increased by each iteration of the while-loop by just one. Thus, if we assume for the sake of contradiction that $i > k$, the condition of the while-loop before the $k$-th iteration must have been true. In particular, we have $\widehat{dist}(\ap_k(u),\ap_k(v))$ exceeding $\frac{1}{4}\widetilde{pivotDist}_{k+1}(\ap_k(v))$ but this gives an immediate contradiction as we define $\widetilde{pivotDist}_{k+1}(x) = \infty$ for all $x \in V(H_i)$ while we assume that $H_k$ is a connected graph.

\paragraph{Lower Bounding the Estimate.} Observe that 
\begin{align*}
    \widetilde{dist}(u,v) &= \sum_{j \leq i} \widetilde{d}_j \geq \sum_{\ell = 1}^{i - 1} \widetilde{pivotDist}_{\ell+1}(\ap_\ell(u)) + dist_{H_i}(\ap_i(u),\ap_i(v)) + \sum_{\ell = 1}^{i - 1} \widetilde{pivotDist}_{\ell+1}(\ap_\ell(v))\\
    &\geq \sum_{\ell = 1}^{i - 1} dist_{H_{\ell}}(\ap_\ell(u), \ap_{\ell+1}(u)) + dist_{H_i}(\ap_i(u),\ap_i(v)) + \sum_{\ell = 1}^{i - 1} dist_{H_{\ell}}(\ap_\ell(v), \ap_{\ell+1}(v))\\
    &\geq \sum_{\ell = 1}^{i - 1} dist_{G}(\ap_\ell(u), \ap_{\ell+1}(u)) + dist_{G}(\ap_i(u),\ap_i(v)) + \sum_{\ell = 1}^{i - 1} dist_{G}(\ap_\ell(v), \ap_{\ell+1}(v))\\
    &\geq dist_G(u,v)
\end{align*}
where we have the first equality from \Cref{lne:addPivotPaths} and \Cref{lne:addFinalPathDist}, the second inequality follows from \Cref{prop:apprxDist}, the third inequality from \Cref{thm:mainApproxTheorem}, and the final inequality from the triangle inequality.

\paragraph{Upper Bounding the Estimate.} We next prove by induction that for each $1 \leq \ell \leq i$, we have $\sum_{j = \ell}^i \widetilde{d}_j \leq 2 \cdot 235885
^{i-\ell} dist_{H_i}(\ap_i(u), \ap_i(v))$. For the base case, we have $\ell = i$, and $\widetilde{d}_{\ell} \leq 2 \cdot dist_{H_{\ell}}(\ap_{\ell}(u), \ap_{\ell}(v))$ by \Cref{thm:algoToMaintainVertexSparsifierHier} which exactly corresponds to the definition in \Cref{lne:addFinalPathDist}. 

We can then take the inductive step for $\ell+1 \mapsto \ell$ for some $\ell < i$. From the while-loop condition in \Cref{lne:whileLoopConditionQuery}, \Cref{thm:algoToMaintainVertexSparsifierHier} and the fact that $\ell < i$, 
\begin{align*}
    dist_{H_{\ell}}(\ap_{\ell}(u),\ap_{\ell}(v))&\geq \frac{1}{8} \max\{ \widetilde{pivotDist}_{\ell+1}(\ap_\ell(u)), \widetilde{pivotDist}_{\ell+1}(\ap_\ell(v)) \} \\ &\geq \frac{1}{8} \max\{ dist_{H_{\ell}}(\ap_{\ell}(u), \ap_{\ell + 1}(u)),  dist_{H_{\ell}}(\ap_{\ell}(v), \ap_{\ell + 1}(v)) \}
\end{align*} 
where the last inequality follows from \Cref{prop:apprxDist}. We further obtain
\begin{align*}
\sum_{j = \ell}^i \widetilde{d}_j &= \sum_{j = \ell + 1}^i \widetilde{d}_j + \widetilde{d}_{\ell} \\
        &\leq 2 \cdot 235885^{i-\ell-1} dist_{H_{\ell+1}}(\ap_{\ell+1}(u), \ap_{\ell+1}(v)) + \widetilde{d}_{\ell}\\
        &\leq 2 \cdot 3629 \cdot 235885^{i-\ell-1} dist_{H_{\ell}}(\ap_{\ell+1}(u), \ap_{\ell+1}(v))  + \widetilde{d}_{\ell}\\
        &\leq 2 \cdot 3629 \cdot 236145^{i-\ell-1} (dist_{H_{\ell}}(\ap_{\ell}(u), \ap_{\ell}(v))\\
        &+ \sum_{x \in \{u,v\}} dist_{H_{\ell}}(\ap_{\ell+1}(x), \ap_{\ell}(x)))  + \widetilde{d}_{\ell}\\
        &\leq 2 \cdot 3629 \cdot 235885^{i-\ell-1} (dist_{H_{\ell}}(\ap_{\ell}(u), \ap_{\ell}(v))\\& + \sum_{x \in \{u,v\}} dist_{H_{\ell}}(\ap_{\ell+1}(x), \ap_{\ell}(x))) \\
        &\leq 2 \cdot 3629 \cdot 235885^{i-\ell-1} \cdot 65 \cdot dist_{H_{\ell}}(\ap_{\ell}(u), \ap_{\ell}(v)) \\
        &= 2 \cdot 235885^{i-\ell} \cdot dist_{H_{\ell}}(\ap_{\ell}(u), \ap_{\ell}(v)) 
\end{align*}
where we use the induction hypothesis in the first inequality, \Cref{thm:mainApproxTheorem} in the second inequality, the triangle inequality in the third inequality. In the fourth inequality, we use that by \Cref{lne:addPivotPaths} and \Cref{prop:apprxDist} of \Cref{def:sparsifier hierarchy}, $\tilde{d}_{\ell} = \widetilde{pivotDist}_{\ell+1}(\ap_\ell(u)) + \widetilde{pivotDist}_{\ell+1}(\ap_\ell(v)) \leq 4 \left(dist_{H_\ell}(\ap_\ell(u), \ap_{\ell+1}(u)) + dist_{H_\ell}(\ap_\ell(v), \ap_{\ell+1}(v))\right)$ and the previous statement. Finally, we use our insight from before. This concludes the induction.

As we have established the claim, and bound $i \leq k$, it suffices to see that this implies in particular that $\widetilde{dist}(u,v) = \sum_{j = 1}^i \widetilde{d}_j \leq 2 \cdot 235885^{i-1} dist_{H_1}(\ap_1(u), \ap_1(v)) = 2 \cdot 235885^{i-1} dist_{G}(u, v) \leq 2 \cdot 235885^{O(\log\log n)} dist_{G}(u, v) \leq \tilde{O}(1) \cdot dist_{G}(u, v)$.

\section{Acknowledgement} We would like to thank anonymous reviewers of STOC for their invaluable feedback.

\printbibliography


\pagebreak

\appendix

\section{Lower Bounding Vertex Sparsifiers}
\label{sec:LowerBoundingVertexSparsifier}

Here, we prove the lower bound given in \Cref{thm:mainApproxTheorem}.

\begin{lemma}
Given a $k$-level hierarchy maintaining algorithm as described in \Cref{def:sparsifier hierarchy}, for any $1 \leq i \leq k$, we have for any $u,v \in V(H_{i})$, $dist_{G}(u,v) \leq dist_{H_{i}}(u,v)$. 
\end{lemma}
\begin{proof}
We prove by induction on level $i$ and stage $t$. The base case for $i = 1$ is straight-forward as $G = H_1$. We can thus take the inductive step $i \mapsto i + 1$ for $1 \leq i < k$. We prove base case and the inductive step for time $t$ simultaneously below, pointing out the difference (if there is any).

Clearly, to establish our claim, it suffices to show that for every edge $e = (x,y) \in E(H_{i+1})$, we have $w(e) \geq dist_{G}(x,y)$. Since there are only 4 types of edges in $E(H_{i+1})$, we can make a simple case analysis:
\begin{itemize}
    \item For an edge $e = (p_{i+1}(u), p_{i+1}(v))$ in $H_{i+1}$ generated by \Cref{iball_edges} for vertices $u,v \in V(H_i)$, we have, by \Cref{prop:apprxDist}, that $dist_{H_i}(p_{i+1}(u), u) \leq \widetilde{pivotDist}_{i+1}(u)$ and 
    \begin{align*}dist_{H_i}(v,p_{i+1}(v)) &\leq 4 \cdot dist_{H_i}(v, V(H_{i+1})) \\&\leq 4 \cdot \left(dist_{H_i}(v,u) + dist_{H_i}(p_{i+1}(u), u)\right) \leq 5 \cdot \widetilde{pivotDist}_{i+1}(u).\end{align*}
    Thus, \begin{align*}
    w(e) &= 8 \cdot \lceil \widetilde{pivotDist}_{i+1}(u) \rceil_{\roundConst} \\&\geq dist_{H_i}(p_{i+1}(u), u) + dist_{H_i}(u,v) + dist_{H_i}(v, p_{i+1}(v)) \\ &\geq dist_{H_{i}}(p_{i+1}(u), p_{i+1}(v)) \geq dist_G(p_{i+1}(u), p_{i+1}(v)).
    \end{align*}
    where in the last inequality, we use the induction hypothesis.

    \item For a generated edge  $e = (\ap^{(t_j)}_{i+1}(v), \ap^{(t_{\ell})}_{i+1}(v))$ in $H_{i+1}$ generated by  \Cref{item:glueEdgesBetweenAPs} for $v \in V(H_i)$ and $0 \leq t_j \leq t_{\ell} \leq t$. By the triangle inequality, \Cref{fact:pivotIsRoughlyMinPivot} and the fact that distances only decrease over time in $H_i$ by its incremental nature, we have
    \begin{align*}
        dist_{H_i}(\ap^{(t_j)}_{i+1}(v), \ap^{(t_\ell)}_{i+1}(v)) &\leq dist_{H_i}(\ap^{(t_j)}_{i+1}(v), v) + dist_{H_i}(v, \ap^{(t_\ell)}_{i+1}(v)) 
        \\&\leq 8 \cdot \lceil \widetilde{minPivotDist}^{(t_j)}_{i+1}(v) \rceil_{\roundConst} = w(e)
    \end{align*}
    and we can finally just use the IH to complete the case.
    
    \item For an edge $e = (\hap^{(t)}_{i+1}(x), \hap_{i+1}^{(t)}(v))$ generated by \Cref{item:glueEdgesBetweenBaseConnectors} for some vertex $v \in V$, time $t' \leq t$ and $x = \hap^{(t')}_{i}(v)$. Note that for time $t = 0$ (the base case), we have that $t' = t$, and therefore $\hap_{i+1}(x) = \hap_{i+1}(v)$ and therefore $e$ is simply a self-loop at $\hap_{i+1}(v)$ which trivially established the claim.
    
    Otherwise, we use that by the triangle inequality, we have $dist_G(\hap^{(t)}_{i+1}(x), \hap_{i+1}^{(t)}(v)) \leq dist_G(\hap^{(t)}_{i+1}(x), x) + dist_G(x,v) + dist_G(v, \hap_{i+1}^{(t)}(v))$. Next, we use for each of these three quantities the following helper fact. 
    
    \begin{fact}\label{fact:helperInLowerBoudnInd}
    For any $y \in V$, $\lceil \widetilde{minPivotDist}_{i+1}(y) \rceil_{\roundConst} \geq dist_G(y, \hap_{i+1}(y))$.
    \end{fact}
    \begin{proof} 
    Let $t'' =  \min \{ t'' \;|\; \widetilde{pivotDist}^{(t'')}_{i+1}(y) \leq  \lceil \widetilde{minPivotDist}_{i+1}(y) \rceil_{\roundConst} \}$. By  \Cref{lastimprPivot}, we have that  $\hap_{i+1}(y) = \ap^{(t'')}_{i+1}(y)$ and we further have that 
    \begin{align*}
        \lceil \widetilde{minPivotDist}_{i+1}(y) \rceil_{\roundConst} &\geq \widetilde{pivotDist}^{(t'')}_{i+1}(y) = \sum_{j \leq i} \widetilde{pivotDist}^{(t'')}_{j+1}(\ap^{(t'')}_j(y)) \\
        &\geq \sum_{j \leq i} dist^{(t'')}_{H_j}(\ap^{(t'')}_j(y), \ap^{(t'')}_{j+1}(y)) \geq \sum_{j \leq i} dist_{G}(\ap^{(t'')}_j(y), \ap^{(t'')}_{j+1}(y)) \\
        &\geq dist_G(y, \hap_{i+1}(y)). 
    \end{align*}
    Here, we first use the definition of $\widetilde{pivotDist}_{i+1}(y)$ from \Cref{prop:apprxDist}, and in the second inequality use the IH and the fact that $G$ itself is incremental. Finally, we use the triangle inequality. 
    \end{proof}
    
    Note in particular, that by invoking the induction hypothesis, we can use  \Cref{fact:helperInLowerBoudnInd} for any time $t'' \leq t$ and level $j$ in lieu of $i$ as long as $j \leq i$. Thus, we can conclude by \Cref{fact:helperInLowerBoudnInd} that $dist_G(x,v) = dist_G(v, \hap^{(t')}_i(v)) \leq \lceil \widetilde{minPivotDist}^{(t')}_{i}(v) \rceil_{\roundConst}$. We thus have 
    \begin{align*}
        dist_G(&\hap^{(t)}_{i+1}(x), \hap_{i+1}^{(t)}(v)) \leq dist_G(\hap^{(t)}_{i+1}(x), x) + dist_G(x,v) + dist_G(v, \hap_{i+1}^{(t)}(v))\\
        &\leq \lceil \widetilde{minPivotDist}^{(t)}_{i+1}(x) \rceil_2 + \lceil \widetilde{minPivotDist}^{(t')}_i(v) \rceil_2 + \lceil \widetilde{minPivotDist}^{(t)}_{i+1}(v) \rceil_2 \\
        &= w(e).
    \end{align*}
    
    \item For an edge $e = (\hap_{i+1}(x), \hap_{i+1}(y))$ generated by \Cref{projected_edges} for some $j \leq i$ and some $e' = (x,y) \in E^{base}(H_j)$. We obtain
     \begin{align*}
        dist_G(\hap_{i+1}(x), \hap_{i+1}(y)) &\leq dist_G(\hap_{i+1}(x), x) + dist_G(x,y) + dist_G(y, \hap_{i+1}(y))\\
        &\leq \lceil \widetilde{minPivotDist}_{i+1}(x)\rceil_{\roundConst} +  dist_G(x,y) + \lceil \widetilde{minPivotDist}_{i+1}(y)\rceil_{\roundConst} \\
        &\leq \lceil \widetilde{minPivotDist}_{i+1}(x)\rceil_{\roundConst} + \lceil w_{H_j}(e) \rceil_{\roundConst} +  \lceil \widetilde{minPivotDist}_{i+1}(y)\rceil_{\roundConst}\\
        &= w(e)
    \end{align*}
    where we first use the triangle inequality, then use \Cref{fact:helperInLowerBoudnInd} and finally use that $dist_G(x,y) \leq w_{H_j}(e)$ by the induction hypothesis.
\end{itemize}
\end{proof}
\end{document}